\documentclass[letterpaper, 11pt, onecolumn, draftclsnofoot]{IEEEtran}
\usepackage{fancyhdr}
\usepackage{amsmath,epsfig}
\usepackage{threeparttable}
\usepackage{epsf,epsfig}
\usepackage{amsthm}
\usepackage{amsmath}
\usepackage{amssymb}
\usepackage{amsfonts}
\usepackage{algorithmic}
\usepackage{algorithm}
\usepackage[noadjust]{cite}
\usepackage{dsfont}
\usepackage{subfigure}

\newtheorem{theorem}{Theorem}

\newtheorem{lemma}{Lemma}
\newtheorem{corollary}{Corollary}

\newtheorem{proposition}{Proposition}

\def\phi{\varphi}

\def\SINR{\mathsf{SINR}}
\def\SNR{\mathsf{SNR}}

\def\l{\left}
\def\r{\right}
\def\({\left(}
\def\){\right)}

\def\bee{{\mathbf{e}}}
\def\bff{{\mathbf{f}}}

\def\bh{{\mathbf{h}}}

\def\bp{{\mathbf{p}}}
\def\bq{{\mathbf{q}}}

\def\bs{{\mathbf{s}}}

\def\bv{{\mathbf{v}}}

\def\b0{{\mathbf{0}}}

\def\cP{\mathcal{P}}

\newcommand{\Pout}{P_{\mathsf{out}}}
\newcommand{\tPout}{\tilde{P}_{\mathsf{out}}}

\newcommand{\nn}{\nonumber}

\begin{document}

\title{\huge \setlength{\baselineskip}{30pt} Cooperative Feedback for Multi-Antenna Cognitive Radio Networks}
\author{\large \setlength{\baselineskip}{15pt} Kaibin Huang and Rui Zhang \thanks{K. Huang is  with Yonsei University, S. Korea. R. Zhang is with National University of Singapore. Email: huangkb@yonsei.ac.kr, elezhang@nus.edu.sg.}\vspace{-40pt}}

\maketitle

\begin{abstract}
Cognitive beamforming (CB) is a multi-antenna technique for efficient spectrum sharing between primary users (PUs) and secondary users (SUs) in a cognitive radio network. Specifically, a multi-antenna SU transmitter applies CB to suppress the interference to the PU receivers as well as  enhance the corresponding  SU-link performance.  In this paper, for a multiple-input-single-output (MISO) SU channel coexisting with a single-input-single-output (SISO) PU channel, we propose a new and practical paradigm for designing CB based on the finite-rate cooperative feedback from the PU receiver to the SU transmitter. Specifically, the PU receiver communicates to the SU transmitter the quantized SU-to-PU channel direction information (CDI) for computing the SU transmit beamformer, and the interference power control (IPC) signal that regulates the SU transmission power according to the tolerable interference margin at the PU receiver. Two  CB algorithms based on cooperative feedback are proposed: one restricts the  SU transmit beamformer  to be orthogonal to the quantized SU-to-PU channel direction and the other relaxes such a constraint. In addition, cooperative feedforward of the SU CDI from the SU transmitter to the PU receiver  is exploited to allow more efficient cooperative feedback. The outage probabilities of the SU link for  different CB  and  cooperative feedback/feedforward algorithms are analyzed, from which the optimal bit-allocation tradeoff between the CDI and IPC feedback is characterized.
\end{abstract}


\section{Introduction}
In a cognitive radio network, secondary users (SUs) are allowed to access the spectrum allocated to a primary network so long as the resultant interference to the primary users (PUs) is within a tolerable margin \cite{Goldsmith2009:BreakingSpectrumGridlockCognitiveRadios}. \emph{Cognitive beamforming}
(CB) is a promising  technique  that enables a  multi-antenna SU transmitter to regulate
its interference to each PU receiver by intelligent beamforming, and thereby transmit more frequently with larger power with respect to a  single-antenna SU transmitter. The optimal CB requires the SU transmitter to acquire  the channel state information (CSI)  of its interference channels to the PU receivers and even that of the primary links, which is difficult without the PUs' cooperation. We consider a two-user cognitive-radio network comprising a multiple-input-single-output (MISO) SU link and a single-input-single-output (SISO) PU link. This paper establishes a new approach of enabling CB at the SU transmitter based on  the finite-rate CSI feedback from the  PU receivers and presents a set of jointly designed CB and  feedback algorithms. The effect of feedback CSI quantization on the SU link performance is quantified, yielding insight into the feedback requirement. 

Existing CB designs assume that the SU transmitter either has prior CSI of the interference channels to the PU receivers or can acquire such information by observing the PU transmissions, which may be  impractical. Assuming perfect CSI of the SU-to-PU channels, the optimal CB
design is proposed in
\cite{Zhang2008:ExploitMultiAntSpectrumSharingCognitiveRadio}
for maximizing the SU throughput  subject to
a given set of interference power constraints at the PU receivers.
The perfect CSI assumption is relaxed in
\cite{Zhang2008:CognitiveBeamformMadePractical} and a more practical
CB algorithm is designed where a SU transmitter estimates the
required CSI by exploiting channel reciprocity and periodically
observing the PU transmissions. However, channel estimation
errors can  cause unacceptable residual
interference from the SU transmitter to the PU receivers. This issue is addressed in \cite{ZhangLiang:RobustCognitiveBeamforming:2008} by optimizing   the cognitive beamfomer to cope with CSI inaccuracy. Besides CB, the power of
the SU transmitter can be adjusted opportunistically to further
increase the SU throughput by exploiting the primary-link CSI as proposed in
\cite{Chen:CognitiveRadioOppPowreCntrl:2008} and
\cite{Zhang2008:OptimalPwrControlCognitiveRadio}. Such CSI, however, is even more
difficult for the SU to obtain than that of the SU-to-PU channels if the PU receivers provide no feedback.

For multiple-input multiple-output 
(MIMO) wireless systems, CSI feedback from the receiver enables precoding at
the transmitter, which not only enhances the  throughput but also simplifies the
transceiver design \cite{ScaStoETAL:OptimalDesignSpaceTimePrecoders:2002}. However, CSI feedback can incur substantial  overhead due
to the multiplicity of MIMO channel coefficients. This motivates
active research on designing efficient feedback quantization
algorithms, called \emph{limited feedback}
\cite{LovHeaETAL:WhatValuLimiFee:Oct:2004}. There exists a rich literature on limited feedback \cite{Love:OverviewLimitFbWirelssComm:2008} where MIMO CSI quantizers have been designed based on various principles such as     line packing
\cite{LovHeaETAL:GrasBeamMultMult:Oct:03} and Lloyd's
algorithm \cite{Lau:MIMOBlockFadingFbLinkCapConst:04}, and targeting different systems ranging from single-user beamforming \cite{MukSabETAL:BeamFiniRateFeed:Oct:03, LovHeaETAL:GrasBeamMultMult:Oct:03} to multiuser downlink 
\cite{Jindal:MIMOBroadcastFiniteRateFeedback:06, SharifHassibi:CapMIMOBroadcastPartSideInfo:Feb:05,
Huang:SDMASumFbRate:06}.  In view of prior work, limited feedback for coexisting networks remains a largely uncharted   area. In particular, there exist  few results on limited feedback for cognitive radio  networks.

In traditional cognitive radio networks, primary users have higher priority of accessing the radio spectrum and are reluctant to cooperate with secondary users having lower priority and belonging  to  an alien network  \cite{Zhao:SurveyDynamSpectAccess:2007}. However, inter-network cooperation is expected in the emerging heterogeneous wireless networks that employ macro-cells, micro-cells, and femto-cells to serve users with different priorities \cite{LTEAdvance}.  
For instance, a macro-cell  mobile user  can assist the cognitive transmission in a nearby  femto-cell. Thus, the design of efficient cooperation methods in cognitive radio networks will facilitate the implementation of next-generation heterogeneous wireless networks. 

This paper presents  a new and practical  paradigm for designing  CB based on  the finite-rate CSI feedback  from  the PU receiver to the SU transmitter, called \emph{cooperative feedback}. To be specific, the PU receiver communicates to
the SU transmitter i) the \emph{channel-direction information} (CDI),
namely the quantized shape of the SU-to-PU MISO channel, for computing the cognitive beamformer and ii)
the \emph{interference-power-control} (IPC) signal that regulates the SU transmission power  according to the tolerable
interference margin at the PU receiver. Our main contributions are summarized as follows.

\begin{enumerate}
\item We present two CB algorithms for the SU
transmitter based on the finite-rate cooperative feedback from the
PU receiver.  One is \emph{orthogonal cognitive beamforming} (OCB) where the SU
transmit beamformer is restricted to be orthogonal to the
feedback SU-to-PU channel shape and the SU
transmission power is controlled  by the IPC feedback. The other is
\emph{non-orthogonal cognitive  beamforming} (NOCB) for which the orthogonality
constraint on OCB is relaxed and the matching IPC signal is designed. 

\item In addition to cooperative feedback, we propose \emph{cooperative feedforward} of
the secondary-link CSI  from the SU transmitter to the PU receiver. The feedforward is found to enable  more efficient IPC feedback, allowing  larger SU transmission power.

\item We analyze the secondary-link performance in terms of the signal-to-noise ratio (SNR) outage
probability for OCB. In particular, regardless of whether there is feedforward, the SU outage
probability  is shown to be lower-bounded in the high SNR regime due to feedback CDI  quantization. The lower bound is proved to decrease exponentially with the number of CDI feedback bits. 

\item Finally, we derive the optimal bit allocation for the CDI and IPC feedback under a sum feedback rate constraint, which minimizes an upper bound on the SU outage probability.
\end{enumerate}

The remainder of this paper is organized as follows. Section~\ref{Section:System} introduces the system model. Section~\ref{Section:Algo} presents the jointly designed CB and cooperative feedback algorithms.  Section~\ref{Section:Outage:Prob} and Section~\ref{Section:Tradeoff} provide  the analysis of  the SU outage probability and the optimal tradeoff between the CDI and IPC feedback-bit allocation, respectively. Simulation results are given in Section~\ref{Section:Simulation}, followed by concluding remarks in  Section~\ref{Section:Conclusion}.

\section{System Model}\label{Section:System}
We consider a primary link coexisting with a secondary link. The
transmitter $\mathsf{T_{\mathsf{p}}}$ and the receiver $\mathsf{R_{\mathsf{p}}}$ of the primary link both
have a single antenna, while the secondary link comprises a
multi-antenna transmitter $\mathsf{T_{\mathsf{s}}}$  and a single-antenna receiver
$\mathsf{R_{\mathsf{s}}}$. 
The multiple antennas at $\mathsf{T_{\mathsf{s}}}$ are employed for beamforming where the beamformer is represented by $\bff$. All channels follow independent block fading. The channel coefficients of the primary and secondary links are independent and identically
distributed (i.i.d.) circularly symmetric complex Gaussian random
variables with zero-mean and unit-variance, denoted by
$\mathcal{CN}(0,1)$. Consequently, the
primary signal received at $\mathsf{R_{\mathsf{p}}}$ has the power $P_{\mathsf{p}} g_{\mathsf{p}}$, where $P_{\mathsf{p}}=\|\bff\|^2$ is
the transmission power of $\mathsf{T_{\mathsf{p}}}$ and $g_{\mathsf{p}}$  the primary channel  power that is  exponentially  distributed
with unit variance, denoted by $\exp(1)$. The MISO channels from
$\mathsf{T_{\mathsf{s}}}$ to $\mathsf{R_{\mathsf{p}}}$ and from $\mathsf{T_{\mathsf{s}}}$ to $\mathsf{R_{\mathsf{s}}}$ are represented by the
$L\times 1$  vectors $\bh_{\textsf{x}}$ consisting of i.i.d.
$\mathcal{CN}(0,1)$ elements and $\bh_{\mathsf{s}}$ comprising i.i.d.
$\mathcal{CN}(0,\lambda)$ elements, respectively, where $0<\lambda <
1$ accounts for a larger path loss  between $\mathsf{T_{\mathsf{s}}}$ and $\mathsf{R_{\mathsf{p}}}$
than that between $\mathsf{T_{\mathsf{s}}}$ and $\mathsf{R_{\mathsf{s}}}$ (or between $\mathsf{T_{\mathsf{p}}}$ and $\mathsf{R_{\mathsf{p}}}$).
To facilitate analysis, $\bh_{\textsf{x}}$ is decomposed into the channel gain
$g_{\textsf{x}} = \|\bh_{\textsf{x}}\|^2/\lambda$ and channel shape $\bs_{\textsf{x}} =
\bh_{\textsf{x}}/\|\bh_{\textsf{x}}\|$, and hence $\bh_{\textsf{x}} = \sqrt{\lambda g_{\textsf{x}}}\bs_{\textsf{x}}$;
similarly, let $\bh_{\mathsf{s}} = \sqrt{g_{\mathsf{s}}}\bs_{\mathsf{s}}$. The channel power $g_{\textsf{x}}$
and $g_{\mathsf{s}}$ follow independent chi-square distributions with $L$
complex degrees of freedom.  

\begin{figure}[t]
\begin{center}
\includegraphics[width=10cm]{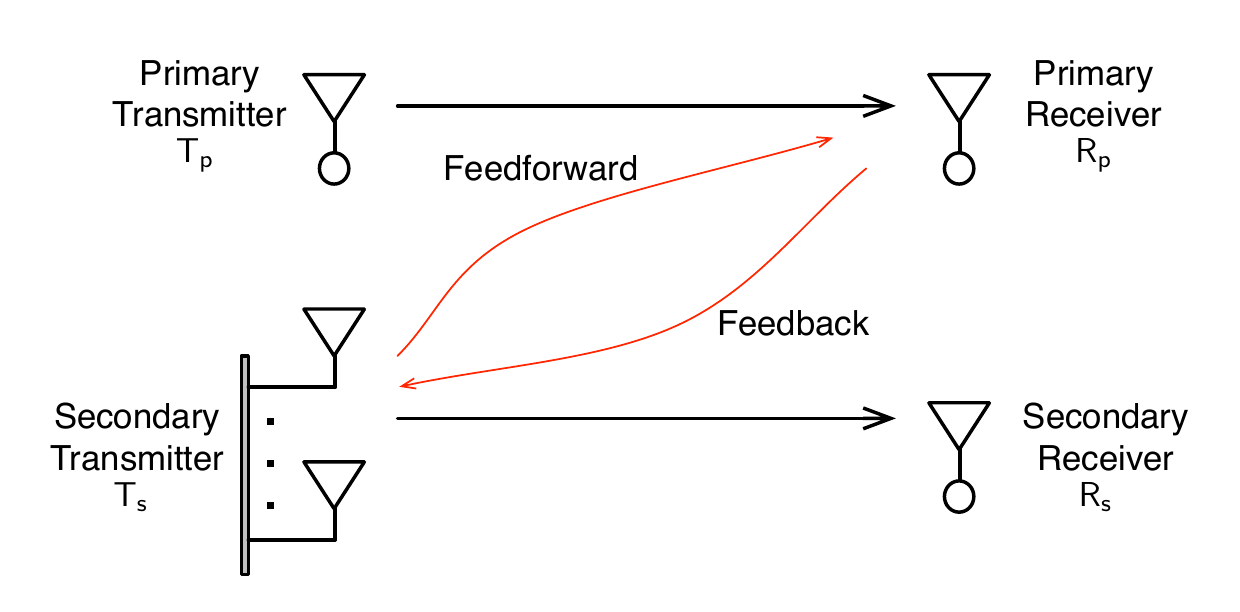}
\caption{Coexisting single-antenna primary and multi-antenna secondary links}
\label{Fig:Sys}
\end{center}
\end{figure}

The primary receiver $\mathsf{R_{\mathsf{p}}}$ cooperates with the secondary transmitter
$\mathsf{T_{\mathsf{s}}}$ to maximize the secondary-link throughput  without compromising the primary-link performance. We assume that $\mathsf{R_{\mathsf{p}}}$ estimates  
$\bh_{\textsf{x}}$ and $g_{\mathsf{p}}$ perfectly and has prior knowledge of  the maximum  SU transmission power 
$P_{\max}$. This enables  $\mathsf{R_{\mathsf{p}}}$ to compute and 
communicate to $\mathsf{T_{\mathsf{s}}}$ the IPC signal
and the CDI $\bs_{\textsf{x}}$. Under a  finite-rate feedback constraint,
the IPC and CDI feedback  must be both quantized. Let $\hat{\bs}_{\textsf{x}}$
denote the output of quantizing $\bs_{\textsf{x}}$. Following
\cite{YooJindal:FiniteRateBroadcastMUDiv:2007,Zhou:QuantifyPowrLossTxBeamFiniteRateFb:2005},
we adopt the quantization model where  $\hat{\bs}_{\textsf{x}}$ lies on  a
hyper sphere-cap centered at $\bs_{\textsf{x}}$ and its radius depends on the
quantization resolution. Specifically, the quantization error
$\epsilon = 1 - |\hat{\bs}_{\textsf{x}}^\dagger\bs_{\textsf{x}}|^2$ has the following cumulative
distribution function  for $L > 1$ 
\cite{YooJindal:FiniteRateBroadcastMUDiv:2007} \footnote{ $\dagger$
denotes   the Hermitian-transpose matrix operation.}
\begin{equation}\label{Eq:CDI:Quant}
\Pr(\epsilon\leq \tau) =  \left\{\begin{aligned}
&2^B \tau^{L-1},&& 0 \leq \tau \leq 2^{-\frac{B}{L-1}}\\
&1,&& \textrm{otherwise}
\end{aligned}\right.
\end{equation}
where $B$ is the number of CDI feedback bits. The  IPC feedback quantization  is discussed in Section~\ref{Section:Algo}.

Feedback of $\bs_{\mathsf{s}}$ from $\mathsf{R_{\mathsf{s}}}$ to $\mathsf{T_{\mathsf{s}}}$ is also required for computing the beamformer $\bff$, called \emph{local feedback}.  
We assume no feedback of $g_{\mathsf{s}}$ from $\mathsf{R_{\mathsf{s}}}$ to $\mathsf{T_{\mathsf{s}}}$. Thus the transmission power $P_{\mathsf{s}}$ of $\mathsf{T_{\mathsf{s}}}$ is independent of $g_{\mathsf{s}}$. 
We also consider the scenario where $\mathsf{T_{\mathsf{s}}}$ sends $\bs_{\mathsf{s}}$ to $\mathsf{R_{\mathsf{p}}}$, called \emph{feedforward},
prior to cooperative  feedback. This information is used by $\mathsf{R_{\mathsf{p}}}$ to
predict the beamformer at $\mathsf{T_{\mathsf{s}}}$ and thereby tolerate larger transmission power at $\mathsf{T_{\mathsf{s}}}$. 
For simplicity, the local  feedback and the feedforward are assumed perfect. This assumption allows us to
focus on the effect of finite-rate cooperative feedback.

The performance of the primary and secondary links are both measured by the SNR or signal-to-interference-plus-noise ratio (SINR)
outage probability. Accordingly, the data rates for
the primary and secondary links are fixed as $R_{\mathsf{p}} =
\log_2(1+\theta_{\mathsf{p}})$ and $R_{\mathsf{s}} = \log_2(1+\theta_{\mathsf{s}})$, respectively,
where $\theta_{\mathsf{p}}$ and $\theta_{\mathsf{s}}$ specify the  receive SNR/SINR thresholds for correct decoding.
The receive SNR and SINR
at $\mathsf{R_{\mathsf{p}}}$ are given by
\begin{equation}
\SNR_{\mathsf{p}} = \gamma_{\mathsf{p}}g_{\mathsf{p}}\quad \textrm{and}\quad \SINR_{\mathsf{p}}  =
\frac{\gamma_{\mathsf{p}}g_{\mathsf{p}}}{1 + \frac{\lambda g_{\textsf{x}}}{\sigma^2}|\bff^{\dagger}\bs_{\textsf{x}}|^2 }
\end{equation}
where $\gamma_{\mathsf{p}}$ is  the PU transmit SNR given by
$\gamma_{\mathsf{p}} = P_{\mathsf{p}}/\sigma^2$, and the noise samples  at both
$\mathsf{R_{\mathsf{p}}}$ and $\mathsf{R_{\mathsf{s}}}$ are i.i.d. $\mathcal{CN}(0,\sigma^2)$ random
variables. The PU outage probability is unaffected by the SU transmission and can be written as 
\begin{eqnarray}
\bar{P}_{\textrm{out}} &=&  \Pr(\SINR_{\mathsf{p}} < \theta_{\mathsf{p}}) \nn\\
&=& \Pr(\SNR_{\mathsf{p}} <  \theta_{\mathsf{p}})\label{Eq:TxConst}\\
&=& 1 - e^{-\frac{\theta_{\mathsf{p}}}{\gamma_{\mathsf{p}}}} \label{Eq:Pout:Prim}
\end{eqnarray}
where the equality in \eqref{Eq:TxConst} specifies  a constraint on the SU CB design and  
\eqref{Eq:Pout:Prim} follows from that the primary channel gain $g_{\mathsf{p}}$ is distributed as $\exp(1)$.
In a heterogeneous network, a primary transmitter such as a macro-cell base station is located far away from a receiver  served by a secondary transmitter such as a femto-cell base station. Therefore, interference from $\mathsf{T_{\mathsf{p}}}$ to $\mathsf{R_{\mathsf{s}}}$ is
assumed negligible and the receive SNR at $\mathsf{R_{\mathsf{s}}}$ is
\begin{equation}\label{Eq:RxSNR:Sec}
\SNR_{\mathsf{s}} = \frac{g_{\mathsf{s}}}{\sigma^2} |\bff^{\dagger} \bs_{\mathsf{s}}|^2. 
\end{equation}
It follows that the SU outage probability is 
\begin{equation}
\Pout = \Pr(\SNR_{\mathsf{s}} \leq \theta_{\mathsf{p}}). 
\end{equation}

\section{Cognitive Beamforming and Cooperative Feedback Algorithms}\label{Section:Algo}
The beamforming algorithms are designed to
minimize the secondary link outage probability under the PU-outage-probability constraint in \eqref{Eq:TxConst}. 
The OCB and NOCB algorithms together with
matching IPC feedback designs are discussed in separate subsections.

\subsection{Orthogonal  Cognitive Beamforming}\label{Section:ZFBeam:Algo}
The OCB beamformer at $\mathsf{T_{\mathsf{s}}}$, denoted as $\bff_{\mathsf{o}}$, suppresses interference
to $\mathsf{R_{\mathsf{p}}}$ and yet enhances $\SNR_{\mathsf{s}}$ in \eqref{Eq:RxSNR:Sec}. To
 this end, $\bff_{\mathsf{o}}$ is constrained to be orthogonal to the
feedback CDI $\hat{\bs}_{\textsf{x}}$, giving the name OCB. Despite the
orthogonality constraint, there exists residual interference from
$\mathsf{T_{\mathsf{s}}}$ to $\mathsf{R_{\mathsf{p}}}$ due to the quantization error in $\hat{\bs}_{\textsf{x}}$. The
interference power can be controlled to satisfy the constraint in
\eqref{Eq:TxConst} using  IPC feedback from $\mathsf{R_{\mathsf{p}}}$ to $\mathsf{T_{\mathsf{s}}}$.
Specifically, the transmission power of $\mathsf{T_{\mathsf{s}}}$, defined as $P_{\mathsf{s}} = \|\bff_{\mathsf{o}}\|^2 $,  satisfies the constraint
$P_{\mathsf{s}}\leq \hat{\eta}$, where $\hat{\eta}$ is the quantized
IPC feedback signal to be designed in the sequel. It follows from above discussion that 
the beamformer  $\bff_{\mathsf{o}}$ solves the following 
optimization problem
\begin{equation}\label{Eq:Beamformeamform}
\begin{aligned}
&\textrm{maximize:} &&|\bff_{\mathsf{o}}^\dagger\bs_{\mathsf{s}}|^2\\
&\textrm{subject to:} &&  \bff_{\mathsf{o}} ^{\dagger}\hat{\bs}_{\textsf{x}}=0 \\
&&& \|\bff_{\mathsf{o}}\|^2 \leq \hat{\eta}.
\end{aligned}
\end{equation}
To solve the above problem,  we decompose  $\bs_{\mathsf{s}}$ as
$\bs_{\mathsf{s}} = a\hat{\bs}_{\textsf{x}} + b\hat{\bs}_{\perp}$ where
$\hat{\bs}_{\perp}$ is an $L\times 1$ vector with unit norm such that   $\hat{\bs}_{\perp}^{\dagger}\hat{\bs}_{\textsf{x}}=0$, and
the coefficients $(a, b)$ satisfy $|a|^2+|b|^2=1$. With this
decomposition, the optimization problem in
\eqref{Eq:Beamformeamform} can be rewritten as
\begin{equation}\label{Eq:Beamformeamform:a}
\begin{aligned}
&\textrm{maximize:}&&|b\bff_{\mathsf{o}}^\dagger\hat{\bs}_{\perp}|^2\\
&\textrm{subject to:}&&\|\bff_{\mathsf{o}}\|^2\leq \hat{\eta}.
\end{aligned}
\end{equation}
It follows that $\bff_{\mathsf{o}}$ implements the maximum-ratio transmission
\cite{Lo:MaxRatioTx:99} and is thus given as
\begin{equation}\label{Eq:Beam:ZF}
\bff_{\mathsf{o}} = \sqrt{\hat{\eta}}\hat{\bs}_{\perp}.
\end{equation}

\subsubsection{The Design of IPC Feedback}\label{Section:IPC:OCB}

The unquantized IPC feedback signal, denoted as
$\eta$, is designed such that the constraint  $\|\bv\|^2 \leq \eta$
is sufficient for enforcing that in \eqref{Eq:TxConst}. The
quantization of $\eta$ will be discussed in the next subsection. The
constraint in \eqref{Eq:TxConst} can be translated into one on the
residual interference power $I_{\mathsf{o}}$ from $\mathsf{T_{\mathsf{s}}}$ to $\mathsf{R_{\mathsf{p}}}$ as follows. Let $\mathsf{null}(\hat{\bs}_{\textsf{x}})$ denote the null space of $\hat{\bs}_{\textsf{x}}$ and its basis vectors are represented as $(\bee_1, \bee_2, \cdots, \bee_{L-1})$. It follows from the CDI quantization model in
Section~\ref{Section:System} that
\begin{equation} \label{Eq:Proj}
\sum_{n=1}^{L-1} |\bs_{\textsf{x}}^\dagger\bee_n|^2 = \epsilon. 
\end{equation}
Without loss of generality, let $\bee_1 = \hat{\bs}_{\perp}$ since $\hat{\bs}_{\perp}\in\mathsf{null}(\hat{\bs}_{\textsf{x}})$  and define $\delta = |\bs_{\textsf{x}}^\dagger \hat{\bs}_{\perp}|^2$. Thus from \eqref{Eq:Proj} and since $|\bs_{\textsf{x}}^\dagger \hat{\bs}_{\perp}|^2 \leq \sum_{n=1}^{L-1} |\bs_{\textsf{x}}^\dagger\bee_n|^2$, we can obtain that $\delta \leq \epsilon$. Furthermore, define $\bq = \sum_{n=2}^{L-1} (\bs_{\textsf{x}}^\dagger\bee_n) \bee_n$ and $\bs_{\textsf{x}}$ can be decomposed as 
\begin{equation}\label{Eq:Ss:Decompose}
\bs_{\textsf{x}} = e^{j\theta_1}\sqrt{1-\epsilon}\hat{\bs}_{\textsf{x}} + e^{j\theta_2}\sqrt{\delta}\hat{\bs}_{\perp} + \bq
\end{equation}
where  the angles
$(\theta_1, \theta_2)$ represent appropriate phase
rotations.  Using the above expression,  $I_{\mathsf{o}}$  can be upper-bounded
as
\begin{eqnarray}
I_{\mathsf{o}} &=& \lambda g_{\textsf{x}} |\bff^\dagger\bs_{\textsf{x}}|^2\label{Eq:Iz:Def}\\
&=&  \lambda g_{\textsf{x}} | \sqrt{\eta}\hat{\bs}_{\perp}^\dagger(e^{j\theta_1}\sqrt{1-\epsilon}\hat{\bs}_{\textsf{x}} + e^{j\theta_2}\sqrt{\delta}\hat{\bs}_{\perp} +\bq)|^2\label{Eq:Iz:a}\\
& =& \lambda g_{\textsf{x}} \eta \delta \label{Eq:Iz}\\
&\leq& \lambda g_{\textsf{x}} \eta\epsilon \label{Eq:I}
\end{eqnarray}
where  \eqref{Eq:Iz:a} is obtained by substituting \eqref{Eq:Beam:ZF} and  \eqref{Eq:Ss:Decompose}. Note that computing
$\delta$  at $\mathsf{R_{\mathsf{p}}}$ requires $\hat{\bs}_{\perp}$ that can be derived from the  feedforward
of $\bs_{\mathsf{s}}$ from $\mathsf{T_{\mathsf{s}}}$. Therefore, $I_{\mathsf{o}}$ can be obtained at $\mathsf{R_{\mathsf{p}}}$
using \eqref{Eq:Iz} for the case of feedforward or otherwise
approximated using \eqref{Eq:I}. Based on the principle of
opportunistic power control  in
\cite{Chen:CognitiveRadioOppPowreCntrl:2008}, the constraint in \eqref{Eq:TxConst} is equivalent to that:
\begin{equation}\label{Eq:I:Const}
I_{\mathsf{o}} \leq  \omega, \quad \textrm{if}\ \omega \geq 0
\end{equation}
where 
\begin{equation}
\omega = \sigma^2\l(\frac{\gamma_{\mathsf{p}}g_{\mathsf{p}}}{\theta_{\mathsf{p}}}-1\r). \label{Eq:Omega}
\end{equation}
If $\omega < 0$,
$I_{\mathsf{o}}$ can be arbitrarily large since $\mathsf{R_{\mathsf{p}}}$ experiences outage even
without any interference from $\mathsf{T_{\mathsf{s}}}$. For the case without
feedforward, the IPC signal $\eta$ is
obtained by combining \eqref{Eq:I} and \eqref{Eq:I:Const} as
\begin{equation}\label{Eq:Ps}
\eta =\l\{\begin{aligned}
&\frac{\omega}{\lambda g_{\textsf{x}} \epsilon},&&\omega \geq 0 \\
&P_{\max}, && \textrm{otherwise}.
\end{aligned}\r.
\end{equation}
The counterpart of $\eta$ for the case of feedforward, denoted as
$\acute{\eta}$, follows from \eqref{Eq:Iz} and \eqref{Eq:I:Const} as
\begin{equation}\label{Eq:Ps:FF}
\acute{\eta} =\l\{\begin{aligned}
&\frac{\omega}{\lambda g_{\textsf{x}} \delta},&&\omega \geq 0 \\
&P_{\max}, && \textrm{otherwise}.
\end{aligned}\r.
\end{equation}
Note that the constraint $P_{\mathsf{s}}\leq \acute{\eta}$ is looser than  $P_{\mathsf{s}}\leq \eta$ 
in the case of $w\geq 0$ since $\delta \leq
\epsilon$.

\subsubsection{The Quantization of IPC Feedback} \label{Section:IPC:OCB:Quant}
Let $\hat{\eta}$
denote the  $(A+1)$-bit output of quantizing $\eta$.  The first bit indicates whether there is an outage event at
$\mathsf{R_{\mathsf{p}}}$; the following $A$ bits represent $\hat{\eta}$ if $\mathsf{R_{\mathsf{p}}}$ is not
in outage (i.e., $\omega \geq 0$) or otherwise are neglected by $\mathsf{T_{\mathsf{s}}}$.
Given $\omega\geq 0$,
$\hat{\eta}$ is constrained to take on values from a finite set  of
$N=2^A$ nonnegative scalars, denoted by $\cP = \{p_0, p_1, \cdots,
p_{N-1}\}$ where $p_0< p_1<\cdots<p_{N-1}$. Note that the optimal
design of $\cP$ for minimizing the SU outage probability
requires additional knowledge at $\mathsf{R_{\mathsf{p}}}$ of the secondary-link  data  rate and
channel distribution. For simplicity, we consider the suboptimal design of $\cP$ whose elements partition the space of $\eta$ using the criterion of
equal probability.\footnote{The IPC quantizer can be improved  by
limiting the quantization range to $P_{\max}$ and optimizing the set
$\mathcal{P}$ using Lloyd's algorithm \cite{Lloyd57}. However, the
corresponding analysis is complicated. Thus, we use the current design for
simplicity and do not pursue the optimization of the IPC
quantization in this work.} Specifically, $p_0 = 0$ and
\begin{equation}\label{Eq:IPCCbk}
\left\{\begin{aligned}
&\Pr(p_n< \eta\leq p_{n+1}\mid  \gamma_{\mathsf{p}} g_{\mathsf{p}}\geq \theta_{\mathsf{p}})  = \frac{1}{N}, && 0\leq n \leq N-2\\
&\Pr(\eta> p_{n}\mid  \gamma_{\mathsf{p}} g_{\mathsf{p}}\geq \theta_{\mathsf{p}})  = \frac{1}{N}, && n = N-1.
\end{aligned}
\r.
\end{equation}
 Given
$\cP$, define the operator $\lfloor \cdot \rfloor_{\mathcal{P}}$ on
$x\geq 0$ as $\lfloor x \rfloor_{\mathcal{P}}  = \max_{p\in
\mathcal{P}} p$ subject to $p \leq x$. Then $\hat{\eta}$ is
given as
\begin{equation}\label{Eq:QIPC}
\hat{\eta} =\l\{\begin{aligned}
&\l\lfloor\eta
\r\rfloor_{\mathcal{P}}, && \omega \geq 0 ~\textrm{and} ~\eta<P_{\max} \\
&P_{\max}, && \textrm{otherwise}.
\end{aligned}
\r.
\end{equation}
Note that $\hat{\eta} \leq \eta$ and thus the constraint
$\|\bff_{\mathsf{o}}\|^2\leq \hat{\eta}$ is sufficient for maintaining the
constraint in \eqref{Eq:I:Const} or its equivalence in \eqref{Eq:TxConst}.
Last, for  the case with feedforward, the output $\tilde{\eta}$ of quantizing $\acute{\eta}$ in \eqref{Eq:Ps:FF} is given by \eqref{Eq:QIPC}  with  $\eta$  replaced  with $\acute{\eta}$.

\subsection{Non-Orthogonal Cognitive Beamforming}

The NOCB beamformer at $\mathsf{T_{\mathsf{s}}}$ is designed by relaxing the
orthogonality constraint on OCB. We formulate the design of
NOCB beamformer as a convex optimization problem and derive its
closed-form solution. The  matching IPC feedback signal is also designed.

The NOCB beamformer, denoted as $\bff_{\mathsf{n}}$, is modified from the OCB
counterpart by replacing the constraint $\bff_{\mathsf{o}}^\dagger\hat{\bs}_{\textsf{x}}
=0$ with $|\bff_{\mathsf{n}}^\dagger \hat{\bs}_{\textsf{x}}|^2 \leq \hat{\mu}_1$ where
$0\leq \hat{\mu}_1 \leq P_{\max}$. In other words, NOCB controls
transmission power in the direction specified by $\hat{\bs}_{\textsf{x}}$ rather
than suppressing it.  In addition, $\bff_{\mathsf{n}}$ satisfies a power
constraint $\|\bff_{\mathsf{n}}\|^2 \leq \hat{\mu}_2$ with $0\leq \hat{\mu}_2
\leq P_{\max}$. The parameters $\hat{\mu}_1$ and $\hat{\mu}_2$
constitute  the quantized IPC feedback signal designed in the
sequel. Under the above constraints, the design of  $\bff_{\mathsf{n}}$ to maximize
the receive SNR at $\mathsf{R_{\mathsf{s}}}$ can be formulated as the following
optimization problem
\begin{equation}
\begin{aligned}
&\textrm{maximize:}&& |\bff_{\mathsf{n}}^\dagger\bs_{\mathsf{s}}|^2\\
&\textrm{subject to:} && |\bff_{\mathsf{n}}^\dagger\hat{\bs}_{\textsf{x}}|^2 \leq \hat{\mu}_1\\
&&&\|\bff_{\mathsf{n}}\|^2 \leq \hat{\mu}_2. 
\end{aligned}\label{Eq:OpBeam:Prob}
\end{equation}
To solve the above problem, we write $\bff_{\mathsf{n}} =
\alpha \hat{\bs}_{\textsf{x}} + \beta\hat{\bs}_{\perp} + \rho\bp$ where $\bp=\bq/\|\bq\|$ with $\bq$  identical to that in \eqref{Eq:Ss:Decompose} 
and $|\alpha|^2 + |\beta|^2 + |\rho|^2 \leq \hat{\mu}_2$.  An optimization  problem having
the same form as \eqref{Eq:OpBeam:Prob} is solved in
\cite{Zhang2008:ExploitMultiAntSpectrumSharingCognitiveRadio}. Using
the results in
\cite[Theorem~2]{Zhang2008:ExploitMultiAntSpectrumSharingCognitiveRadio}
and $\bs_{\mathsf{s}} = a\hat{\bs}_{\textsf{x}} + b\hat{\bs}_{\perp}$, we obtain the
following lemma.
\begin{lemma}\label{Lem:NOBeam} The NOCB beamformer is given by $\bff_{\mathsf{n}} = \alpha \hat{\bs}_{\textsf{x}} + \beta\hat{\bs}_{\perp}$ where
\begin{itemize}
\item[--] If $\hat{\mu}_1 \geq |a|^2\hat{\mu}_2$
\begin{equation}
\alpha = a \sqrt{\hat{\mu}_2},\quad \beta = b\sqrt{\hat{\mu}_2}\label{Eq:Beam:Case1}
\end{equation}

\item[--] If $0\leq \hat{\mu}_1  < |a|^2\hat{\mu}_2$
\begin{equation}
\alpha = a \sqrt{\hat{\mu}_1},\quad \beta = b\sqrt{\hat{\mu}_2-\hat{\mu}_1}.
\end{equation}
\end{itemize}
\end{lemma}
Note that the beamformer in \eqref{Eq:Beam:Case1} performs the maximum-ratio transmission \cite{Lo:MaxRatioTx:99}.

In the remainder of this section, the IPC feedback signal $\hat{\mu}
= (\hat{\mu}_1, \hat{\mu}_2)$ is designed to enforce the constraint
in \eqref{Eq:TxConst}.  The unquantized version of $\hat{\mu}$, denoted as $\mu = (\mu_1,
\mu_2)$, is first designed as follows. Similar to
\eqref{Eq:I:Const}, the constraint  in \eqref{Eq:TxConst} can be
transformed into the following constraint on the residual
interference power $I_{\mathsf{n}}$ from $\mathsf{T_{\mathsf{s}}}$ to $\mathsf{R_{\mathsf{p}}}$:
\begin{eqnarray}
I_{\mathsf{n}} &=& \lambda g_{\textsf{x}} |\bff_{\mathsf{n}}^\dagger \bs_{\textsf{x}}|^2  \nn\\
&\leq& \omega, \quad \textrm{if}\ \omega \geq 0 \label{Eq:I:Const:NB}
\end{eqnarray}
or otherwise $\|\bff_{\mathsf{n}}\|^2 = P_{\max}$. To facilitate the design, $I_{\mathsf{n}}$ is upper-bounded as
follows:
\begin{eqnarray}
I_{\mathsf{n}}
&=& \lambda g_{\textsf{x}} |(\alpha\hat{\bs}_{\textsf{x}} + \beta\hat{\bs}_{\perp})^\dagger (e^{j\theta_1} \sqrt{1-\epsilon}\hat{\bs}_{\textsf{x}} + e^{j\theta_2}\sqrt{\delta}\hat{\bs}_{\perp} + \bq)|^2\label{Eq:1}\\
&=& \lambda g_{\textsf{x}} |\alpha e^{j\theta_1}\sqrt{1-\epsilon} + \beta e^{j\theta_2}\sqrt{\delta} |^2\nn\\
&\leq& \lambda g_{\textsf{x}} (|\alpha| \sqrt{1-\epsilon} + |\beta|\sqrt{\delta} )^2\nn\\
&\leq& \lambda g_{\textsf{x}} (|\alpha| \sqrt{1-\epsilon} + \sqrt{P_{\max}\delta} )^2\label{Eq:Io:Ub:a}\\
&\leq& \lambda g_{\textsf{x}} (|\alpha|\sqrt{1-\epsilon} + \sqrt{P_{\max}\epsilon } )^2\label{Eq:Io:Ub}
\end{eqnarray}
where \eqref{Eq:1} uses Lemma~\ref{Lem:NOBeam} and \eqref{Eq:Ss:Decompose},
\eqref{Eq:Io:Ub:a} applies $|\beta|^2 \leq
P_{\max}$, and \eqref{Eq:Io:Ub} follows from $\delta \leq \epsilon$. Recall
that computing $\delta$ at $\mathsf{R_{\mathsf{p}}}$ requires  feedforward.
Therefore, for the case without feedforward, the bound on $I_{\mathsf{n}}$ in
\eqref{Eq:Io:Ub} should be used in
designing the IPC feedback. Specifically, combining \eqref{Eq:I:Const:NB} and \eqref{Eq:Io:Ub}
 gives the following constraint on $\alpha$
\begin{equation}\label{Eq:Const:Alpha}
|\alpha| \leq \nu,\quad \textrm{if}\ \nu \geq 0
\end{equation}
where
\begin{equation}
\nu= \frac{\sqrt{\frac{\omega}{\lambda g_{\textsf{x}}}}-\sqrt{\epsilon P_{\max}}}{\sqrt{1-\epsilon}}. \label{Eq:Nu}
\end{equation}
For $\nu \geq 0$, it follows that $\mu_1 = \nu^2$. For $\nu<0$, the
above constraint is invalid and thus we set $\mu_1 = |\alpha|^2=0$; as
a result, the NOCB optimization problem in \eqref{Eq:OpBeam:Prob}
converges to the OCB counterpart in \eqref{Eq:Beamformeamform},
leading to $\mu_2 = \eta$. Furthermore, it can be observed from
\eqref{Eq:Io:Ub:a} that setting $\mu_2 = P_{\max}$ for the case of
$\mu_1 >0$ does not violate the interference constraint in
\eqref{Eq:I:Const:NB}. Combining above results gives the following  IPC
feedback design: 
\begin{equation}\label{Eq:IPC:NOCB}
\mu = \l\{\begin{aligned}
&(\nu^2, P_{\max}), && \nu \geq 0,  \omega \geq 0 \\
&(0, \eta), && \nu < 0, \omega \geq 0 \\
&(P_{\max}, P_{\max}), && \omega < 0
\end{aligned}
\r.
\end{equation}
where $\eta$ and $\nu$ are given in  \eqref{Eq:Ps} and \eqref{Eq:Nu}, respectively. 
It follows that the quantized IPC feedback, denoted as $\hat{\mu}$, is given as
\begin{equation}\label{Eq:IPC:NOCB:Quant}
\hat{\mu} = \l\{\begin{aligned}
&(\hat{\nu}, P_{\max}), && \nu \geq 0, \omega \geq 0 \\
&(0, \hat{\eta}), && \nu < 0, \omega \geq 0 \\
&(P_{\max}, P_{\max}), && \omega < 0
\end{aligned}
\r.
\end{equation}
where $\hat{\nu}=\lfloor\nu^2\rfloor_{\mathcal{P}'}$ with
$\mathcal{P}'$ being a scalar quantizer codebook designed similarly as $\mathcal{P}$ discussed in
Section~\ref{Section:IPC:OCB:Quant}.  The feedback of $\hat{\mu}$ is
observed from \eqref{Eq:IPC:NOCB:Quant} to involve the transmission
of only  a single  scalar (either $\hat{\nu}$ or $\hat{\eta}$) with one
additional bit for separating the first two cases in
\eqref{Eq:IPC:NOCB:Quant}. Note that the third case can be
represented by setting $\hat{\nu}=P_{\max}$.

For the case with feedforward, the IPC feedback
is designed by applying  the constraint in \eqref{Eq:I:Const:NB} to the
upper bound on $I_{\mathsf{n}}$  in \eqref{Eq:Io:Ub:a} and following 
similar steps as discussed earlier. The resultant
quantized IPC feedback, denoted as $\check{\mu}$, is
\begin{equation}\label{Eq:IPC:NOCB:FF}
\check{\mu} = \l\{\begin{aligned}
&(\check{\nu}, P_{\max}), && \acute{\nu} \geq 0, \omega \geq 0 \\
&(0, \tilde{\eta}), && \acute{\nu} < 0, \omega \geq 0 \\
&(P_{\max}, P_{\max}), && \omega < 0
\end{aligned}
\r.
\end{equation}
where  $\check{\nu}=
\lfloor\acute{\nu}^2\rfloor_{\check{\mathcal{P}}}$ with $\acute{\nu}$
being   the unquantized IPC feedback signal
 \begin{equation}\label{Eq:Const:Alpha new}
\acute{\nu} = \frac{\sqrt{\frac{\omega}{\lambda g_{\textsf{x}}}}-\sqrt{\delta
P_{\max}}}{\sqrt{1-\epsilon}}
\end{equation}
and  $\check{\mathcal{P}}$ a suitable quantizer codebook. Note that
$\check{\mu} \geq \hat{\mu}$ since $\check{\nu}\geq \hat{\nu}$ and
$\tilde{\eta}\geq \hat{\eta}$. In other words, feedforward relaxes the constraint on the SU transmission power. 

 \subsection{Comparison between Orthogonal and Non-Orthogonal Cognitive Beamforming}
Regardless of whether feedforward exists, NOCB outperforms  OCB since NOCB relaxes the SU transmission  power constraint with respect to OCB, which can be verified by  comparing the IPC signals in \eqref{Eq:Ps} and \eqref{Eq:Ps:FF} with those  in  \eqref{Eq:IPC:NOCB}  and \eqref{Eq:IPC:NOCB:FF}, respectively. Next, the performance of OCB and NOCB converges as
$P_{\max}\rightarrow \infty$.  Let $\Pout$ and $\tPout$ denote the SU outage probabilities for OCB and NOCB, respectively. 
\begin{proposition}\label{Cor:Converge}
For large $P_{\max}$, the SU outage probabilities for OCB and NOCB converge as
\begin{equation}\label{Eq:App:i}
\lim_{P_{\max}\rightarrow \infty} \tPout = \lim_{P_{\max}\rightarrow \infty} \Pout
\end{equation}
regardless of whether feedforward is available. 
\end{proposition}
\begin{proof}
% Remove from the final submission
See Appendix~\ref{App:Converge}.
\end{proof}
The above discussion is consistent with  simulation results in
Fig.~\ref{Fig:OptimBeam}.

\subsection{The Effect of Quantizing  Local Feedback and Feedforward}

In practice, the local feedback and feedforward of  $\bs_{\mathsf{s}}$ must be quantized like the cooperative feedback signals. Let $\hat{\bs}_{\mathsf{s}}$ denote the quantized version of $\bs_{\mathsf{s}}$. The corresponding cognitive beamforming and cooperative feedback algorithms can be modified from  those in the  preceding sections by replacing $\bs_{\mathsf{s}}$ with $\hat{\bs}_{\mathsf{s}}$. The error in the feedback/feedforward of $\bs_{\mathsf{s}}$ at most causes a loss  on the received SNR at $\mathsf{R_{\mathsf{s}}}$ without  affecting  the primary link performance, which does not change the fundamental results of this work. Note that  extensive work has been carried out on quantifying the performance  loss of beamforming systems caused by  local feedback quantization  (see e.g., \cite{Zhou:QuantifyPowrLossTxBeamFiniteRateFb:2005, LovHeaETAL:GrasBeamMultMult:Oct:03,MukSabETAL:BeamFiniRateFeed:Oct:03}). Furthermore, simulation results presented in Fig.~\ref{Fig:QuantFeedforward} confirm that the quantization of $\bs_{\mathsf{s}}$ has insignificant effect on the SU outage probability, justifying  the current assumption of perfect local feedback and feedforward.

\section{Outage Probability}\label{Section:Outage:Prob}
The CDI typically requires more feedback bits than the IPC signal  since the former
is  an   $L\times 1$ complex vector and the latter is a real
scalar. For this reason, assuming perfect IPC feedback, this section
focuses on quantifying  the effects of CDI quantization on the SU
outage probability for OCB. Similar  analysis  for
NOCB is  complicated with little new insight and
hence omitted.

\subsection{Orthogonal  Cognitive Beamforming without Feedforward}
The outage  probability depends on the distribution of the SU transmission power $P_{\mathsf{s}}$, which is given in the
following lemma.

\begin{lemma}\label{Lem:Ps}
For OCB without feedforward, the distribution of $P_{\mathsf{s}}$ is
given as
\begin{eqnarray}
\Pr(P_{\mathsf{s}} = P_{\max}) &\!=\!& 1- e^{-\frac{\theta_{\mathsf{p}}}{\gamma_{\mathsf{p}}}} \l[\frac{(L-1)\theta_{\mathsf{p}}\lambda \gamma_{\max}}{\gamma_{\mathsf{p}}}\r]2^{-\frac{B}{L-1}}+O\l(2^{-\frac{2B}{L-1}}\r)\label{Eq:P1}\\
\Pr(P_{\mathsf{s}} < \tau)&\!=\!& e^{-\frac{\theta_{\mathsf{p}}}{\gamma_{\mathsf{p}}}}
\l[\frac{(L-1)\theta_{\mathsf{p}}\lambda\tau}{\gamma_{\mathsf{p}}\sigma^2}\r]2^{-\frac{B}{L-1}}+O\l(2^{-\frac{2B}{L-1}}\r),
~\forall ~0\leq\tau\leq P_{\max}\label{Eq:P2}
\end{eqnarray}
where $\gamma_{\max} = P_{\max}/\sigma^2$.
\end{lemma}
\begin{proof}
See Appendix~\ref{App:Ps}.
\end{proof}
For a sanity check, from the above results,
\begin{equation}
\lim_{B\rightarrow\infty}\Pr(P_{\mathsf{s}} = P_{\max}) = 1\quad
\textrm{and}\quad \lim_{B\rightarrow\infty}\Pr(P_{\mathsf{s}}<P_{\max})=0. \nn
\end{equation}
These are consistent  with the fact that OCB with perfect CDI feedback
($B\rightarrow\infty$) nulls the interference from $\mathsf{T_{\mathsf{s}}}$ to $\mathsf{R_{\mathsf{p}}}$, allowing $\mathsf{T_{\mathsf{s}}}$ to always transmit
using the maximum power.

Next, define the effective channel power of the secondary link as
$\tilde{g}_{\mathsf{s}} = g_{\mathsf{s}}|\hat{\bff}^{\dagger}_{\mathsf{o}} \bs_{\mathsf{s}}|^2$ with
$\hat{\bff}_{\mathsf{o}}=\bff_{\mathsf{o}}/\sqrt{P_{\mathsf{s}}}$. The following result directly follows from \cite[Lemma~2]{Jindal:RethinkMIMONetwork:LinearThroughput:2008} on
zero-forcing beamforming for mobile ad hoc networks.

\begin{lemma} \label{Lem:gs}
The effective channel power $\tilde{g}_{\mathsf{s}}$  is a chi-square random
variable with $(L-1)$ complex degrees of freedom, whose probability density function is given
as
\begin{equation}
f_{\tilde{g}_{\mathsf{s}}}(\tau) = \frac{\tau^{L-2}}{\Gamma(L-1)}e^{-\tau}
\end{equation}
where $\Gamma(\cdot)$ denotes the gamma function.
\end{lemma}

Using Lemmas~\ref{Lem:Ps} and \ref{Lem:gs}, the main result of this
section is obtained as shown in the following theorem.
\begin{theorem}\label{Theo:Pout:PerfPwrFb}
The SU outage probability for OCB without feedforward is
\begin{equation}\label{Eq:Pout}
\Pout =1 - \frac{\Gamma\l(L-1,  \frac{\theta_{\mathsf{s}}}{\gamma_{\max}}\r)}{\Gamma(L-1)} +\phi 2^{-\frac{B}{L-1}} + O\l(2^{-\frac{2B}{L-1}}\r)
\end{equation}
where $\Gamma(\cdot, \cdot)$ denote the incomplete gamma function and 
\begin{equation}
\phi =e^{-\frac{\theta_{\mathsf{p}}}{\gamma_{\mathsf{p}}}} \frac{(L-1)\lambda\theta_{\mathsf{p}}\theta_{\mathsf{s}}\Gamma\l(L-2,  \frac{\theta_{\mathsf{s}}}{\gamma_{\max}}\r)}{\gamma_{\mathsf{p}}\Gamma(L-1)}.
\end{equation}
\end{theorem}
\begin{proof}See Appendix~\ref{App:Outage}.
\end{proof}
The last two terms in \eqref{Eq:Pout} represent the increase of the
SU outage probability due to the feedback CDI quantization. The asymptotic
outage probabilities for large $P_{\max}$ and  $B$
are  given in the following two corollaries.

\begin{corollary}\label{Cor:Pout:HiSNR}
For large $P_{\max}$, the SU outage probability in Theorem
\ref{Theo:Pout:PerfPwrFb} converges as
\begin{eqnarray}
\lim_{P_{\max}\rightarrow\infty}\Pout &=& e^{-\frac{\theta_{\mathsf{p}}}{\gamma_{\mathsf{p}}}}\frac{(L-1)\lambda\theta_{\mathsf{p}}\theta_{\mathsf{s}}}{(L-2)\gamma_{\mathsf{p}}}2^{-\frac{B}{L-1}}+ O\l(2^{-\frac{2B}{L-1}}\r)\label{Eq:Pout:HiSNR:a}\\
 &> & e^{-\frac{\theta_{\mathsf{p}}}{\gamma_{\mathsf{p}}}}\frac{\lambda\theta_{\mathsf{p}}\theta_{\mathsf{s}}}{\gamma_{\mathsf{p}}}2^{-\frac{B}{L-1}}+ O\l(2^{-\frac{2B}{L-1}}\r)\label{Eq:Pout:HiSNR:b}.
\end{eqnarray}
\end{corollary}

This result in \eqref{Eq:Pout:HiSNR:a} shows that for large $P_{\max}$,
$\Pout$ saturates at a  level that depends   on the
quality of CDI feedback because the transmission by $\mathsf{T_{\mathsf{s}}}$
contributes residual interference to $\mathsf{R_{\mathsf{p}}}$. The saturation
level of $\Pout$  in \eqref{Eq:Pout:HiSNR:a} decreases exponentially
with increasing $B$, which suppresses   the  residual interference. More details can be found in Fig.~\ref{Fig:Pout} and
the related discussion in Section \ref{Section:Simulation}.

To facilitate subsequent discussion, we refer to the range of
$P_{\max}$ where  $\Pout$ saturates as the \emph{interference
limiting regime}. From \eqref{Eq:Pout:HiSNR:b}, it can be observed that in the
interference limiting regime $\Pout$ increases with the number of
antennas $L$. The reason is that the CDI quantization error
grows with $L$ if $B$ is fixed, thus increasing the  residual
interference from $\mathsf{T_{\mathsf{s}}}$ to $\mathsf{R_{\mathsf{p}}}$. To prevent $\Pout$ from growing
with $L$ in the interference limiting regime, $B$ has to increase at
least linearly with $(L-1)$. However, $\Pout$ decreases with $L$
outside the interference limiting regime, as shown by simulation
results in Fig. \ref{Fig:Feedforward} in Section~\ref{Section:Simulation}.

\begin{corollary}\label{Cor:Pout:Perfect CDI}
For large $B$, the SU outage probability in Theorem
\ref{Theo:Pout:PerfPwrFb} converges as 
\begin{equation}
\lim_{B\rightarrow\infty}\Pout = 1 - \frac{\Gamma\l(L-1,
\frac{\theta_{\mathsf{s}}}{\gamma_{\max}}\r)}{\Gamma(L-1)}.\label{Eq:2}
\end{equation}
\end{corollary}
As $B\rightarrow\infty$, both links are decoupled  and the limit of $\Pout$ in \eqref{Eq:2} decreases continuously with
$P_{\max}$. 

\subsection{Orthogonal Cognitive  Beamforming with Feedforward}\label{Section:ZFFF:Outage}
Effectively, feedforward changes the
analysis in the preceding section by replacing $\epsilon$ with
$\delta$. 
\begin{lemma}\label{Lem:PDF:EspP} The probability density function of $\delta$ is given as
\begin{equation}
f_{\delta}(\tau) =(L-1)2^{\frac{B}{L-1}}\l(1- 2^{\frac{B}{L-1}}\tau\r)^{L-2},\quad 0\leq \tau \leq 2^{-\frac{B}{L-1}}.
\end{equation}
\end{lemma}
\begin{proof} See Appendix~\ref{App:PDF:EspP}.
\end{proof}

Let $\acute{P}_{\mathsf{s}}$  represent the transmission power of $\mathsf{T_{\mathsf{s}}}$ for the
case of feedforward. 
\begin{lemma}\label{Lem:Ps:FF}
The distribution  of $\acute{P}_{\mathsf{s}}$ is given as
\begin{eqnarray}
\Pr(\acute{P}_{\mathsf{s}} = P_{\max}) &\!=\!& 1- e^{-\frac{\theta_{\mathsf{p}}}{\gamma_{\mathsf{p}}}} \l(\frac{\theta_{\mathsf{p}}\lambda \gamma_{\max}}{\gamma_{\mathsf{p}}}\r)2^{-\frac{B}{L-1}}+O\l(2^{-\frac{2B}{L-1}}\r)\label{Eq:P1:FF}\\
\Pr(\acute{P}_{\mathsf{s}} < \tau)&\!=\!& e^{-\frac{\theta_{\mathsf{p}}}{\gamma_{\mathsf{p}}}}
\l(\frac{\theta_{\mathsf{p}}\lambda\tau}{\gamma_{\mathsf{p}}\sigma^2}\r)2^{-\frac{B}{L-1}}+O\l(2^{-\frac{2B}{L-1}}\r),
~ \forall ~ 0\leq \tau\leq P_{\max}. \label{Eq:P2:FF}
\end{eqnarray}
\end{lemma}
\begin{proof}See Appendix~\ref{App:Ps:FF}.
\end{proof}

The following theorem is proved using Lemma~\ref{Lem:Ps:FF} and
following the same procedure as 
Theorem~\ref{Theo:Pout:PerfPwrFb}.

\begin{theorem}\label{Theo:Pout:PerfPwrFb:FF}
For the case of OCB with feedforward, the SU outage probability is
\begin{equation}\label{Eq:Pout:FF}
\acute{P}_{\mathsf{out}} =1 - \frac{\Gamma\l(L-1,  \frac{\theta_{\mathsf{s}}}{\gamma_{\max}}\r)}{\Gamma(L-1)} +\frac{\phi}{L-1} 2^{-\frac{B}{L-1}} + O\l(2^{-\frac{2B}{L-1}}\r)
\end{equation}
where $\phi$ is given in Theorem~\ref{Theo:Pout:PerfPwrFb}.
\end{theorem}
By comparing Theorems~\ref{Theo:Pout:PerfPwrFb} and
\ref{Theo:Pout:PerfPwrFb:FF}, it can be observed that feedforward
reduces the increment of the outage probability due to feedback-CDI
quantization by a factor of $(L-1)$. Thus, the outage probability
reduction with feedforward is more significant for larger  $L$  as confirmed by simulation results (see Fig.~\ref{Fig:Feedforward} in Section~\ref{Section:Simulation}).

\section{Tradeoff between IPC and CDI Feedback}\label{Section:Tradeoff}
In this section, we consider both quantized CDI and IPC feedback.
Using results derived in the preceding section and under a sum
feedback rate constraint, the optimal  allocation of bits to the IPC and
CDI feedback is derived for  OCB.

First, consider OCB without feedforward. Let $\hat{P}_{\mathsf{s}}$ denote the transmission power of $\mathsf{T_{\mathsf{s}}}$. 
The loss on $\hat{P}_{\mathsf{s}}$ due to the IPC feedback quantization is bounded
by a function of the number of IPC feedback bits $A$. Define the
index $1\leq n_0\leq (2^A-1)$ such that $p_{n_0-1} \leq P_{\max}\leq
p_{n_0}$ where $p_n \in \mathcal{P}$. Then the IPC
power loss $(P_{\mathsf{s}}-\hat{P}_{\mathsf{s}})$ can be upper bounded by $\Delta P$ defined as:
\begin{equation}\label{Eq:Delta:P}
\Delta P = \max_{1\leq n \leq n_0}(p_n - p_{n-1}).
\end{equation}

\begin{lemma} \label{Lem:Delta:P} $\Delta P$ defined in \eqref{Eq:Delta:P} is given by
\begin{equation}
\Delta P = \frac{\gamma_{\mathsf{p}}\sigma^2}{(L-1)\theta_{\mathsf{p}}\lambda}2^{\frac{B}{L-1}-A} + O\l(2^{-\frac{B}{L-1}}\r).
\end{equation}
\end{lemma}
\begin{proof} See Appendix~\ref{App:Delta:P}.
\end{proof}

Next, the cumulative distribution function of $\hat{P}_{\mathsf{s}}$ is upper-bounded as shown below.
\begin{lemma}\label{Lem:Ps:QuantPower}
The distribution  of $\hat{P}_{\mathsf{s}}$ satisfies
\begin{eqnarray}
\Pr(\hat{P}_{\mathsf{s}}=P_{\max}) &=& \Pr(P_{\mathsf{s}}=P_{\max})\label{Eq:QIPC:Pmax}\\
\Pr(\hat{P}_{\mathsf{s}} < \tau)&\leq& e^{-\frac{\theta_{\mathsf{p}}}{\gamma_{\mathsf{p}}}}
\l[\frac{(L-1)\theta_{\mathsf{p}}\lambda(\tau+\Delta
P)}{\gamma_{\mathsf{p}}\sigma^2}\r]2^{-\frac{B}{L-1}}
+O\l(2^{-\frac{2B}{L-1}}\r), ~\forall ~ 0\leq \tau\leq P_{\max}
\label{Eq:QIPC:CDF}
\end{eqnarray}
where $\Pr(P_{\mathsf{s}}=P_{\max})$ and  $\Delta P$ are given in Lemma~\ref{Lem:Ps} and Lemma~\ref{Lem:Delta:P}, respectively.
\end{lemma}
\begin{proof}
See Appendix~\ref{App:Ps:QuantPower}.
\end{proof}

Using Lemma~\ref{Lem:Ps:QuantPower} and following the  procedure for
proving Theorem~\ref{Theo:Pout:PerfPwrFb}, the outage probability
for OCB without feedforward  is bounded as shown below.
\begin{proposition}\label{Prop:Pout:QuantPwr}
Given both quantized CDI and IPC feedback, the SU outage
probability for  OCB without feedforward satisfies
\begin{equation}\label{Eq:Pout:QuantPower}
\hat{P}_{\mathsf{out}} \leq 1 - \frac{\Gamma\l(L-1,  \frac{\theta_{\mathsf{s}}}{\gamma_{\max}}\r)}{\Gamma(L-1)} +\phi 2^{-\frac{B}{L-1}}+\alpha 2^{-A}  + O\l(2^{-\frac{2B}{L-1}}\r)
\end{equation}
where $\phi$ is given in Theorem~\ref{Theo:Pout:PerfPwrFb} and
$\alpha =e^{-\frac{\theta_{\mathsf{p}}}{\gamma_{\mathsf{p}}}} \frac{\Gamma\l(L-1,  \frac{\theta_{\mathsf{s}}}{\gamma_{\max}}\r)}{\Gamma(L-1)}$. 
\end{proposition}
Comparing the above result with  \eqref{Eq:Pout}, the increment of 
$\Pout$ due to IPC feedback quantization is upper-bounded by the
term $\alpha2^{-A}$. The asymptotic result parallel to that in
Corollary~\ref{Cor:Pout:HiSNR} is given below.
\begin{corollary}\label{Cor:Pout:HiSNR new}
For large $P_{\max}$, the upper bound on the SU outage probability
$\hat{P}_{\mathsf{out}}$ converges as
\begin{eqnarray}\label{Eq:Pout:Asymp:QIPC}
\lim_{P_{\max}\rightarrow\infty}\hat{P}_{\mathsf{out}} \leq \phi'
2^{-\frac{B}{L-1}}+\alpha' 2^{-A}\label{Eq:Pout:HiSNR:QIPC}
\end{eqnarray}
where $\phi' = e^{-\frac{\theta_{\mathsf{p}}}{\gamma_{\mathsf{p}}}}\frac{(L-1)\lambda\theta_{\mathsf{p}}\theta_{\mathsf{s}}}{(L-2)\gamma_{\mathsf{p}}}$ and $\alpha' =e^{-\frac{\theta_{\mathsf{p}}}{\gamma_{\mathsf{p}}}} $.
\end{corollary}
The two terms at the right-hand side of \eqref{Eq:Pout:Asymp:QIPC}
quantify  the effects of CDI and IPC quantization, respectively.
The exponent of the first term, namely $-\frac{B}{L-1}$, is scaled
by the factor $\frac{1}{L-1}$, which does not appear in that of the
second term. The reason is that the CDI quantization partitions the
space of $L$-dimensional  unitary vectors while the IPC
quantization discretizes the nonnegative real axis.

Consider the sum-feedback constraint $A+B=F$. Note that  $(F+1)$
represents the total number of feedback bits where the additional
bit is used as an indicator of an outage event at $\mathsf{R_{\mathsf{p}}}$. 
Assume $B\gg 1$ and the second-order term
$O\l(2^{-\frac{2B}{L-1}}\r)$ in \eqref{Eq:Pout:QuantPower} is
negligible. Then  the optimal value of
$B$ that minimizes  the upper bound on $\hat{P}_{\mathsf{out}}$ in \eqref{Eq:QIPC:CDF}, denoted as $B^\star$,   is obtained as
\begin{equation}\label{Eq:Tradeoff:Prob}
B^\star = \arg\min_{0\leq B\leq F} J(B)
\end{equation}
where the function $J(B)$ is defined as 
\begin{equation}\label{Eq:Tradeoff:J}
J(B) = \phi 2^{-\frac{B}{L-1}}+\alpha2^{-(F-B)}.
\end{equation}
The function $J(B)$ can be shown to be convex. Thus, by relaxing the
integer constraint, $B^\star$ can be computed using the following
equation
\begin{equation}
\frac{dJ}{dB}(B^\star) = -\frac{\ln 2\times \phi}{L-1}2^{-\frac{B^\star}{L-1}} + \ln 2\times \alpha 2^{-F} 2^{B^\star} = 0.
\end{equation}
It follows that
\noindent \begin{equation}\label{Eq:B:Optim}
B^\star = \min\l[\frac{L-1}{L}\l(F-\log_2\chi\r)^+, F\r]
\end{equation}
where
\begin{equation}
\chi = \frac{\gamma_{\mathsf{p}}\Gamma(L-1, \frac{\theta_{\mathsf{s}}}{\gamma_{\max}})}{\lambda\theta_{\mathsf{p}}\theta_{\mathsf{s}}\Gamma(L-2, \frac{\theta_{\mathsf{s}}}{\gamma_{\max}})}\nn
\end{equation}
and the operator $(\cdot)^+$ is defined as $(\cdot)^+ = \max(\cdot, 0)$. 
The value of $B^\star$ as computed above can then be rounded to
satisfy the integer constraint. The derivation of \eqref{Eq:B:Optim}
uses the first-order approximation of the upper bound on
$\hat{P}_{\mathsf{out}}$  in Proposition~\ref{Prop:Pout:QuantPwr},
which is accurate for relatively small value of
$\frac{\theta_{\mathsf{p}}\theta_{\mathsf{s}}}{\gamma_{\mathsf{p}}}2^{-\frac{B}{L-1}}$. In this range, the feedback allocation
using \eqref{Eq:B:Optim} closely  predicts  the optimal feedback
tradeoff as observed from simulation results in
Fig.~\ref{Fig:Tradeoff} in Section~\ref{Section:Simulation}.  However, the mentioned first-order approximation  is inaccurate for large $\theta_{\mathsf{p}}\theta_{\mathsf{s}}$ or small 
$\gamma_{\mathsf{p}}$. For these cases,  it is necessary to derive the optimal
feedback allocation based on analyzing the exact distribution of 
$\hat{P}_{\mathsf{out}}$, which, however, has no simple form.

Next, consider OCB with feedforward. The feedforward counterpart of
Proposition~\ref{Prop:Pout:QuantPwr} is obtained as the following corollary.
\begin{corollary}\label{Cor:Pout:QuantPwr}
Given both quantized CDI and IPC feedback, the SU outage probability for
 OCB with feedforward satisfies
\begin{equation}\label{Eq:Pout:QuantPower:FF}
\acute{P}_{\mathsf{out}} \leq 1 - \frac{\Gamma\l(L-1,  \frac{\theta_{\mathsf{s}}}{\gamma_{\max}}\r)}{\Gamma(L-1)} + \frac{\phi+\alpha 2^{-A}}{L-1}  2^{-\frac{B}{L-1}} + O\l(2^{-\frac{2B}{L-1}}\r)
\end{equation}
\end{corollary}
The result in \eqref{Eq:Pout:QuantPower:FF} shows that feedforward
reduces the increment on outage probability due to IPC quantization
by a factor of $(L-1)$. Since the solution of the optimization
problem in \eqref{Eq:Tradeoff:Prob} also minimizes the upper bound on
$\acute{P}_{\mathsf{out}}$ in \eqref{Eq:Pout:QuantPower:FF}, the
optimal number of CDI feedback bits in \eqref{Eq:B:Optim} holds for the case with feedforward.

\section{Simulation Results }\label{Section:Simulation}
Unless specified otherwise, the simulation parameters  are set as: the SINR/SNR thresholds $\theta_{\mathsf{p}}=\theta_{\mathsf{s}} = 3$, the
path-loss factor $\lambda = 0.1$, the noise
variance $\sigma^2=1$, the 
number of antennas at $\mathsf{T_{\mathsf{s}}}$ $L = 4$, and the PU transmit SNR 
$\gamma_{\mathsf{p}} = 10$ dB. All curves in the following figures are obtained by simulation except for the curve with the legend ``Quantized CDI (Infinite SNR)" in Fig~\ref{Fig:Pout}, which is based on numerical computation using \eqref{Eq:Pout:HiSNR:a}. 
\begin{figure}[t]
\begin{center}
\includegraphics[width=10cm]{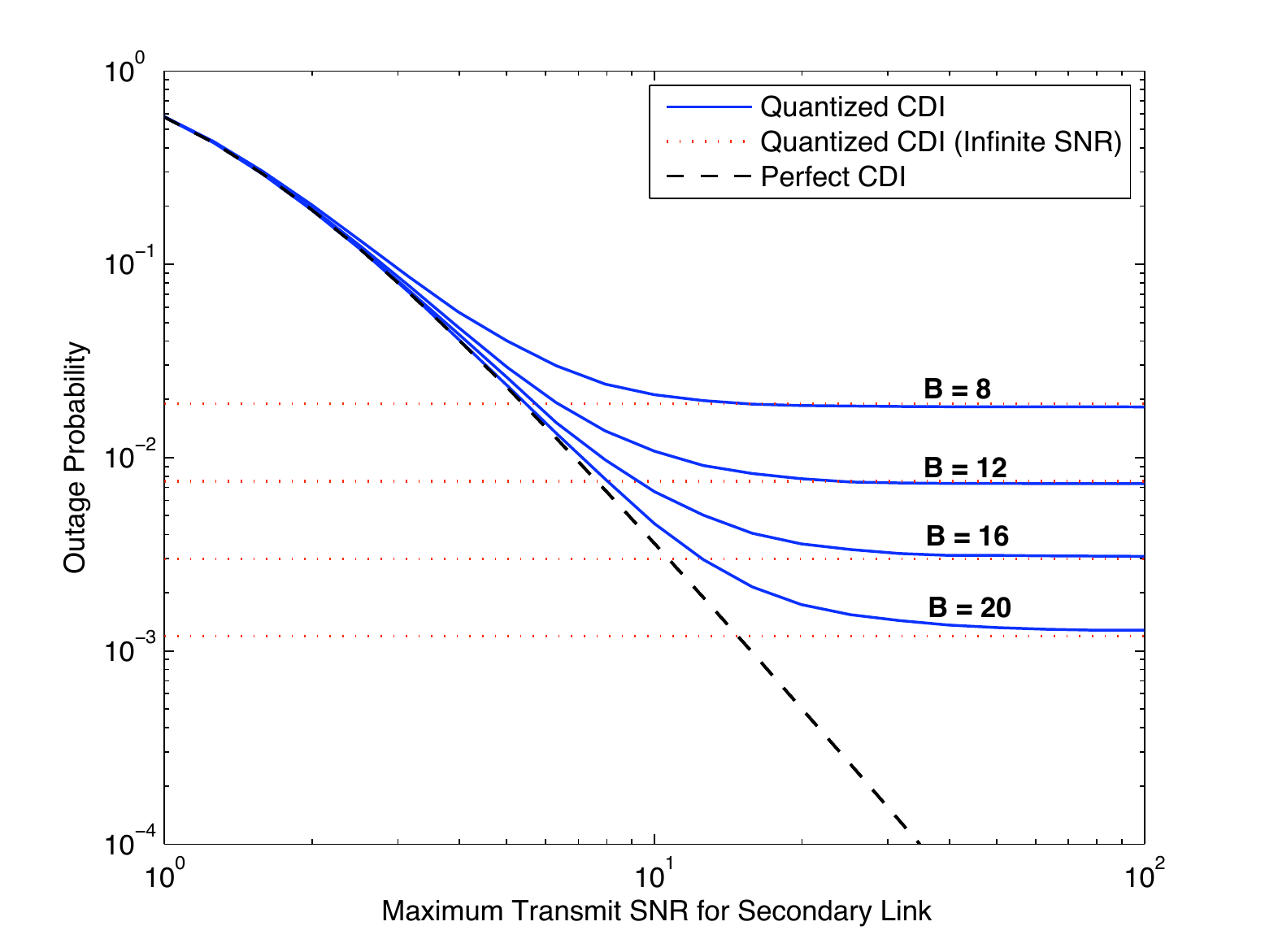}
\caption{SU outage probability for OCB versus maximum SU-transmit SNR for quantized CDI and perfect IPC feedback}
\label{Fig:Pout}
\end{center}
\end{figure}

Figs.~\ref{Fig:Pout} to Fig.~\ref{Fig:QuantFeedforward} concern OCB with quantized CDI and perfect IPC feedback.
Fig.~\ref{Fig:Pout} displays the curves of SU
outage probability $\Pout$ versus maximum SU transmit SNR
$\gamma_{\max}$  for the number of cooperative CDI feedback bits $B = \{8,12,
16, 20\}$.   For comparison, we also plot the first-order terms of the $\Pout$ limits for large $\gamma_{\max}$  as given in
\eqref{Eq:Pout:HiSNR:a}. As observed from
Fig.~\ref{Fig:Pout}, with $B$ fixed, $\Pout$ converges from above to
the corresponding limit  as $\gamma_{\max}$ increases, consistent with 
the result in Corollary~\ref{Cor:Pout:HiSNR}. The limit of $\Pout$
in the interference limiting regime is observed to decrease
exponentially with increasing $B$.

Fig.~\ref{Fig:OptimBeam} shows that  the SU outage probabilities of OCB and NOCB converge 
as $\gamma_{\max}$ increases, agreeing with  Proposition ~\ref{Cor:Converge}. The convergence is
slower for larger $B$. However,  NOCB significantly outperforms  OCB outside the interference limiting  regime.
\begin{figure}
\begin{center}
\includegraphics[width=10cm]{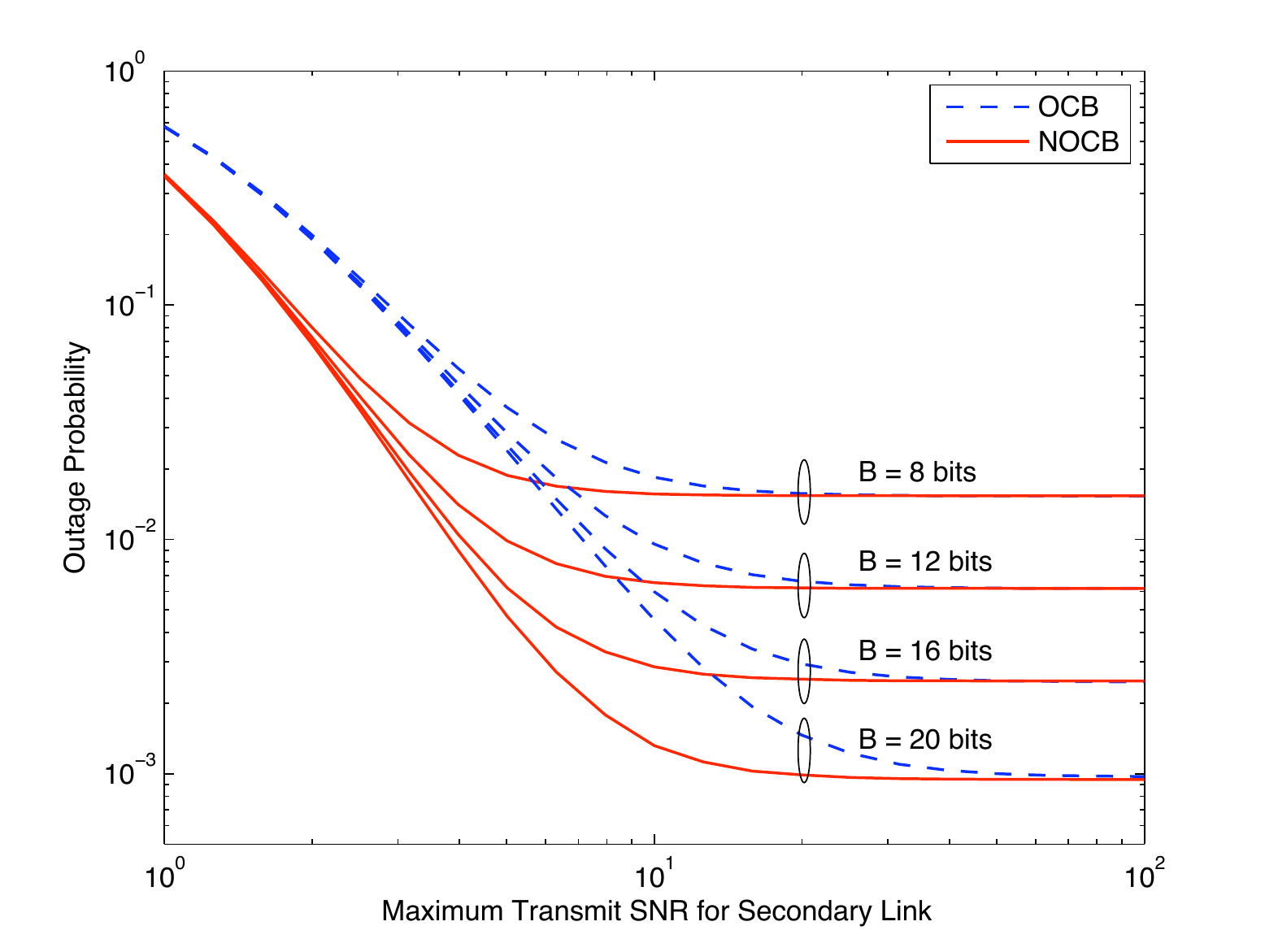}
\caption{Performance comparison between OCB and NOCB in terms of SU outage probability  versus maximum SU transmit SNR. The CDI feedback is quantized  and the IPC feedback is assumed perfect. }
\label{Fig:OptimBeam}
\end{center}
\end{figure}

Fig.~\ref{Fig:Feedforward} illustrates the effects of 
the SU feedforward on  the SU outage probability.  The
OCB and NOCB beamforming designs are considered in
Fig.~\ref{Fig:Feedforward}(a) and Fig.~\ref{Fig:Feedforward}(b),
respectively. 
It can be observed from both figures that the decrease of the SU outage probability due to feedforward is
more significant in the interference limiting regime and for larger
$L$. However, increasing $L$ is found to result in higher outage
probability in the interference limiting regime, which agrees with the remark on Corollary~\ref{Cor:Pout:HiSNR}. 

\begin{figure}
\begin{center}
\hspace{-20pt}\subfigure[OCB]{\includegraphics[width=8.6cm]{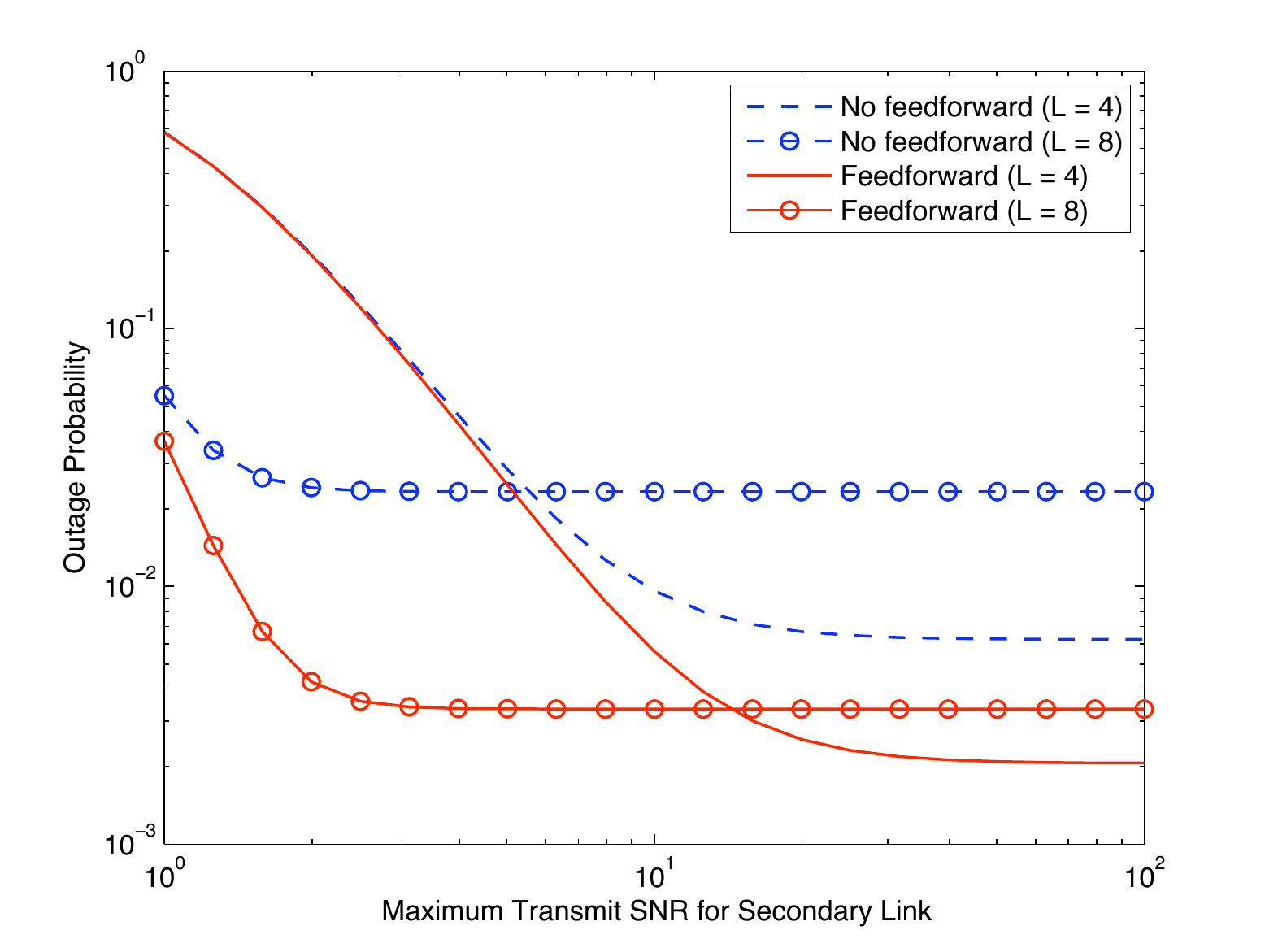}\hspace{-20pt}}
\subfigure[NOCB]{\includegraphics[width=8.6cm]{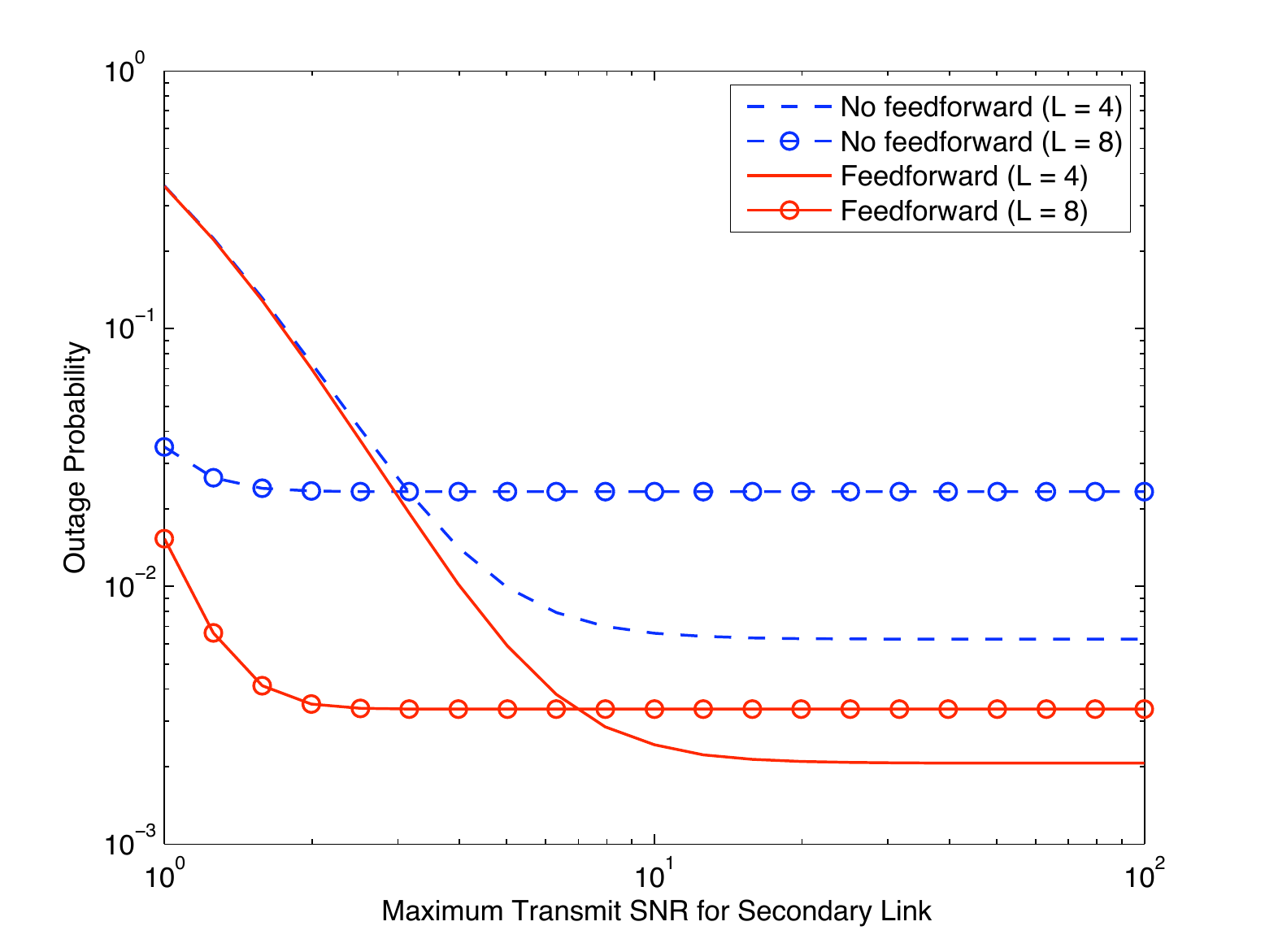}}\hspace{-20pt}
\caption{SU outage probability  versus the maximum SU transmit SNR  for (a) OCB and (b) NOCB. For each type of beamforming, both the cases of feedforward and no feedforward are considered. The number of cooperative CDI feedback bits is $B = 12$. }
\label{Fig:Feedforward}
\end{center}
\end{figure}

Fig.~\ref{Fig:QuantFeedforward} demonstrates the effects of finite-rate ($B'$ bits) local feedback and feedforward of $\bs_{\mathsf{s}}$ on the SU outage probability for OCB with $B = 12$ and $B' = 8$.  It can be observed for both OCB and NOCB that the increase of the SU outage probability due to the quantization of $\bs_{\mathsf{s}}$ is insignificant, justifying the assumption of perfect local feedback and feedforward in the analysis. 

\begin{figure}
\begin{center}
\hspace{-20pt}\subfigure[OCB]{\includegraphics[width=8.5cm]{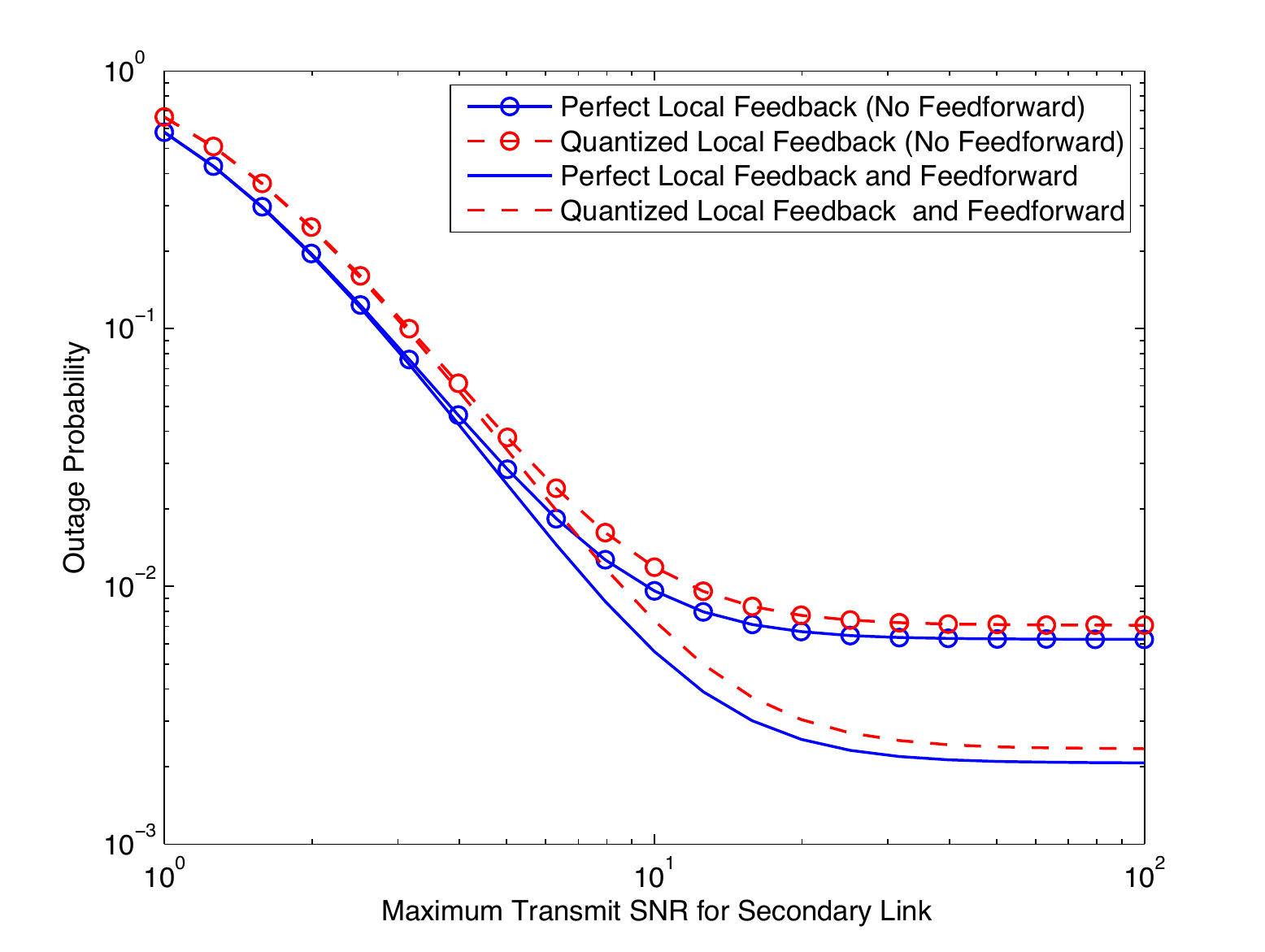}\hspace{-20pt}}
\subfigure[NOCB]{\includegraphics[width=8.5cm]{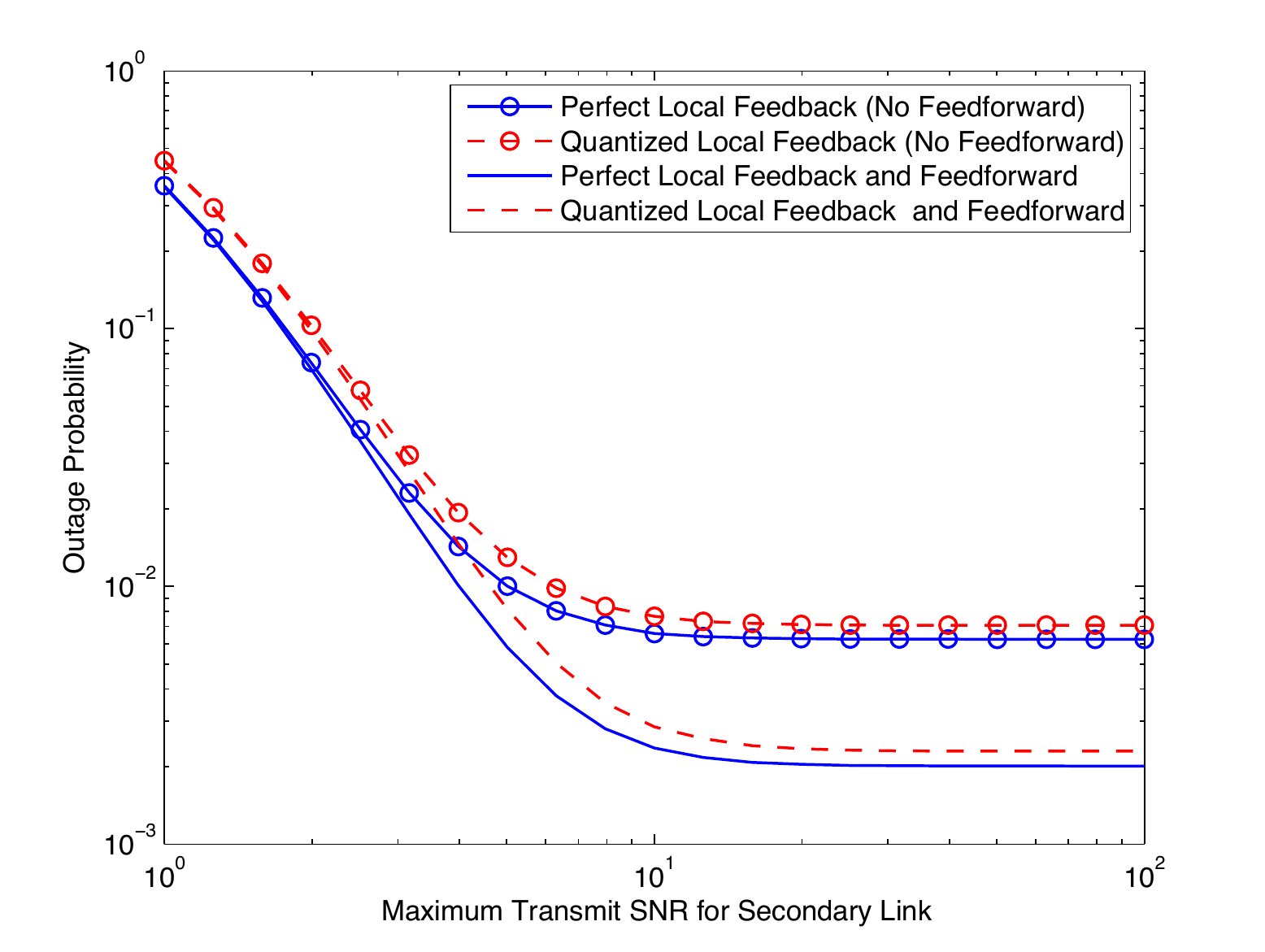}\hspace{-20pt}}
\caption{Effects of quantizing feedforward and local feedback on the SU outage probability    for (a) OCB and (b) NOCB. The number of cooperative CDI feedback bits is $B = 12$, and the local feedback/feedforward has  $B' = 8$ bits. }
\label{Fig:QuantFeedforward}
\end{center}
\end{figure}

Last, consider  both quantized CDI and IPC feedback.
Fig.~\ref{Fig:Tradeoff} shows the curves of the SU outage probability 
for the case of OCB without feedforward versus the number of bits $A$
for IPC feedback. It is observed from Fig.~\ref{Fig:Tradeoff} that for  given
$\gamma_{\max}$,  there exists an optimal combination of $(A,B)$
that  minimizes the outage probability. The optimal values of $A$
are indicated by the marker ``o" and those computed using the
theoretic result in \eqref{Eq:B:Optim} by the marker ``x". The simulation and theoretic results are closer for larger $\gamma_{\mathsf{p}}$. Specifically, they differ at most by two  bits for $\gamma_{\mathsf{p}} = 10$ dB (see  Fig.~\ref{Fig:Tradeoff}(a)) and by one bit for $\gamma_{\mathsf{p}} = 13$ dB (see  Fig.~\ref{Fig:Tradeoff}(b)).  
 These observations  agree with the remark in Section~\ref{Section:Tradeoff} that 
the derived feedback tradeoff is a more accurate approximation of the optimal one for larger $\gamma_{\mathsf{p}}$.
\begin{figure}
\begin{center}
\hspace{-20pt}\subfigure[$\gamma_{\mathsf{p}}=10$ dB]{\includegraphics[width=8.5cm]{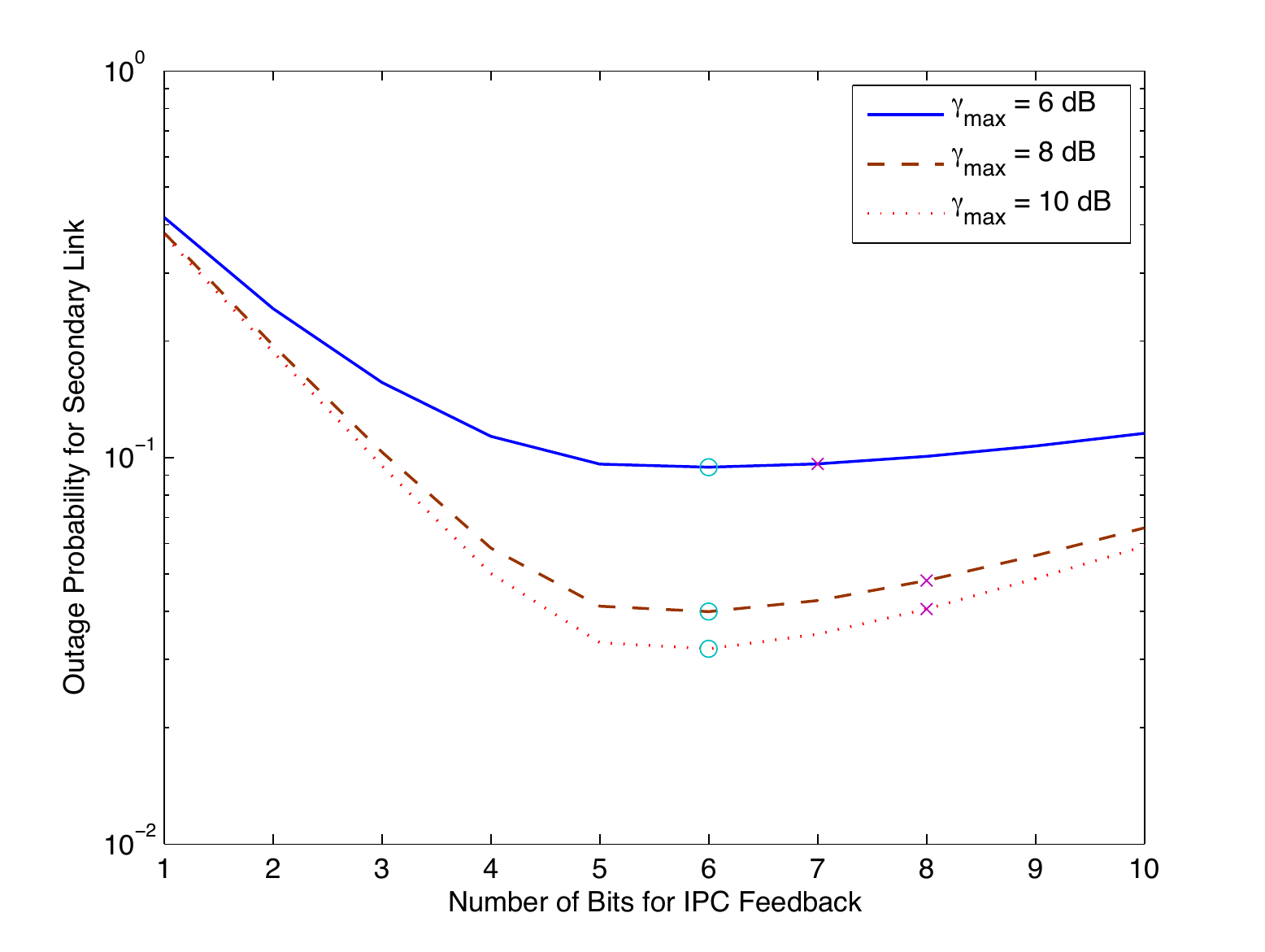}\hspace{-20pt}}
\subfigure[$\gamma_{\mathsf{p}}=13$ dB]{\includegraphics[width=8.5cm]{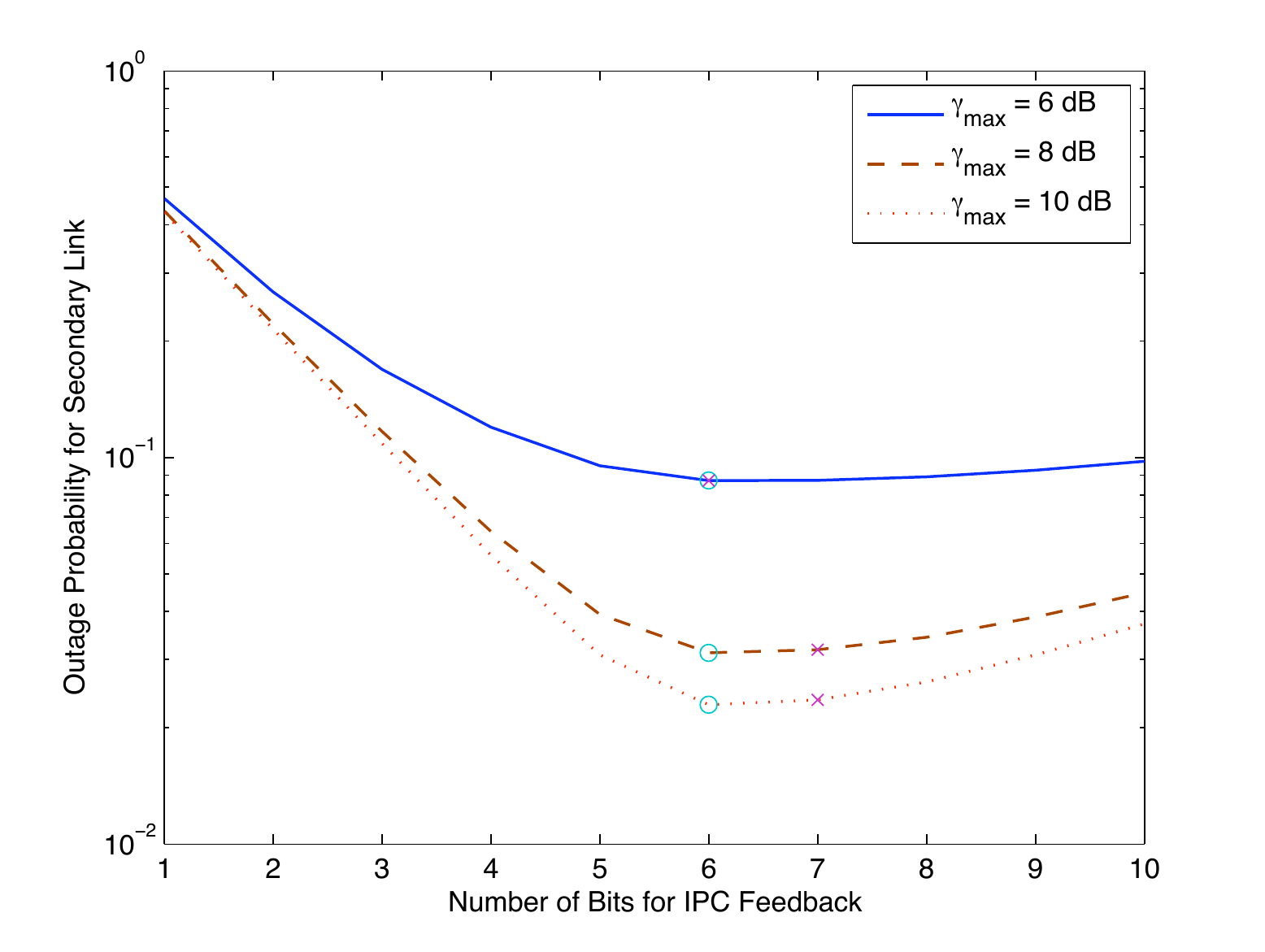}\hspace{-20pt}}\\
\caption{SU outage probability  versus the number of quantized IPC feedback bits for OCB without feedforward. The PU transmit SNR (a) $\gamma_{\mathsf{p}} = 10$ dB and (a) $\gamma_{\mathsf{p}} = 13$ dB. The total number of bits for  CDI and IPC feedback is $A+B = 12$.}
\label{Fig:Tradeoff}
\end{center}
\end{figure}

\section{Concluding Remarks}\label{Section:Conclusion}
We have introduced a new operation model for coexisting PU
and SU links in a spectrum sharing network, where the PU receiver
cooperatively  feeds back quantized side information to the SU
transmitter for facilitating its opportunistic transmission, such
that the resultant PU link performance degradation is minimized. Furthermore, based on cooperative feedback, we have proposed two algorithms for  the SU transmit beamforming  to improve the SU-link
performance. 
Under a  PU-feedback-rate constraint, we have derived  the optimal feedback bits allocation for the CDI and IPC feedback. 
In addition, we have shown  that additional cooperative
feedforward of the SU CDI from the SU transmitter to the PU receiver further enhances the SU-link performance. 

To the authors' best knowledge, this paper is the first attempt in
the literature to study the design of cooperative feedback from the
PU to the SU in a cognitive radio network. This work opens several issues 
worth further investigation. This paper has assumed single antennas for both PU and SU receivers. It is interesting to
extend the proposed  CB and cooperate feedback schemes  to the more general case with MIMO PU and SU links.
Moreover,  we have assumed a single SU link coexisting with a single PU link, while it is pertinent to investigate the more
general system model with multiple coexisting PU and SU links.

\appendix
\subsection{Proof of Proposition~\ref{Cor:Converge}}\label{App:Converge}

For the case without feedforward,  we can expand  $\tPout$ and  $\Pout$ as
\begin{eqnarray}
\tPout &=& \Pr(|\bff_{\mathsf{n}}^\dagger\bh_{\mathsf{s}}|^2 < \theta_{\mathsf{s}}\sigma^2)\nn\\
&=&\l[\Pr(|\bff_{\mathsf{n}}^\dagger\bh_{\mathsf{s}}|^2 < \theta_{\mathsf{s}}\sigma^2\mid \hat{\mu}_1\geq 0 , \omega \geq 0)\Pr(\hat{\mu}_1\geq 0 \mid \omega \geq 0)\r.+\nn\\
&& \l. \Pr(|\bff_{\mathsf{n}}^\dagger\bh_{\mathsf{s}}|^2 < \theta_{\mathsf{s}}\sigma^2\mid \hat{\mu}_1=0, \omega \geq 0)\Pr(\hat{\mu}_1=0\mid \omega \geq 0)\r]\Pr( \omega \geq 0)+\nn\\
 &&\Pr(|\bff_{\mathsf{n}}^\dagger\bh_{\mathsf{s}}|^2 < \theta_{\mathsf{s}}\sigma^2\mid \omega < 0)\Pr( \omega < 0)  \label{Eq:App:j}\\
 \Pout&=& \Pr(|\bff_{\mathsf{o}}^\dagger\bh_{\mathsf{s}}|^2 < \theta_{\mathsf{s}}\sigma^2)\nn\\
 &=&\Pr(|\bff_{\mathsf{o}}^\dagger\bh_{\mathsf{s}}|^2 < \theta_{\mathsf{s}}\sigma^2\mid \omega \geq 0)\Pr( \omega \geq 0)+ \Pr(|\bff_{\mathsf{o}}^\dagger\bh_{\mathsf{s}}|^2 < \theta_{\mathsf{s}}\sigma^2\mid \omega < 0)\Pr( \omega < 0)  \label{Eq:App:o} 
\end{eqnarray}
where $\omega$ is defined in \eqref{Eq:Omega} and $\hat{\mu}_1$ is the IPC feedback parameter in \eqref{Eq:IPC:NOCB:Quant}. 
Using \eqref{Eq:Const:Alpha} and \eqref{Eq:IPC:NOCB:Quant}
\begin{eqnarray}
\lim_{P_{\max}\rightarrow \infty} \Pr(\hat{\mu}_1=0\mid \omega \geq 0) & =&\lim_{P_{\max}\rightarrow \infty} \Pr(\nu < 0 \mid \omega \geq 0) \nn\\
&=& 1. \label{Eq:App:g}
\end{eqnarray}
Moreover, from Lemma~\ref{Lem:NOBeam} and  \eqref{Eq:IPC:NOCB:Quant}
\begin{eqnarray}
\lim_{P_{\max}\rightarrow \infty} \Pr(|\bff_{\mathsf{n}}^\dagger\bh_{\mathsf{s}}|^2 < \theta_{\mathsf{s}}\sigma^2\mid \omega < 0) &=& \lim_{P_{\max}\rightarrow \infty}\Pr\l(P_{\max} < \frac{\theta_{\mathsf{s}}\sigma^2}{g_{\mathsf{s}}}\r)\nn\\
& =&0.\label{Eq:App:h}
\end{eqnarray}
Similarly, it can be shown that 
\begin{eqnarray}
\lim_{P_{\max}\rightarrow \infty} \Pr(|\bff_{\mathsf{o}}^\dagger\bh_{\mathsf{s}}|^2 < \theta_{\mathsf{s}}\sigma^2\mid \omega < 0) & =&0.\label{Eq:App:m}
\end{eqnarray}
By combining \eqref{Eq:App:j}, \eqref{Eq:App:g}, and \eqref{Eq:App:h}
\begin{eqnarray}
\lim_{P_{\max}\rightarrow \infty}\tPout &=&\lim_{P_{\max}\rightarrow \infty}  \Pr(|\bff_{\mathsf{n}}^\dagger\bh_{\mathsf{s}}|^2 < \theta_{\mathsf{s}}\sigma^2\mid \hat{\mu}_1=0, \omega \geq 0)\Pr( \omega \geq 0)\nn\\
&=&\lim_{P_{\max}\rightarrow \infty}  \Pr(|\bff_{\mathsf{o}}^\dagger\bh_{\mathsf{s}}|^2 < \theta_{\mathsf{s}}\sigma^2\mid\omega \geq 0)\Pr( \omega \geq 0)\label{Eq:App:k}
\end{eqnarray}
where \eqref{Eq:App:k} holds since  $\bff_{\mathsf{o}} = \bff_{\mathsf{n}}$ for $\hat{\mu}_1 = 0$. 
Combining \eqref{Eq:App:o}, \eqref{Eq:App:k} and \eqref{Eq:App:m} gives the desired result for the case without feedforward. The proof for the case with feedforward is similar and omitted for brevity. 

\subsection{Proof of Lemma~\ref{Lem:Ps}}\label{App:Ps}
Given perfect IPC feedback and the OCB design specified by \eqref{Eq:Beam:ZF} and \eqref{Eq:Ps}
\begin{eqnarray}\label{Eq:P1:a}
 \Pr(P_{\mathsf{s}} = P_{\max}) &=& \bar{P}_{\textsf{out}} +\Pr\( \frac{\omega}{\lambda g_{\textsf{x}}\epsilon}\geq P_{\max} \r)
\end{eqnarray}
where $\bar{P}_{\textsf{out}}$ is given in \eqref{Eq:Pout:Prim}.
From the definition of $\omega$ in \eqref{Eq:Omega} and define $u = \frac{\theta_{\mathsf{p}}\lambda \gamma_{\max}}{\gamma_{\mathsf{p}}}$, the last term in \eqref{Eq:P1:a} can be obtained as 
\begin{eqnarray}
\Pr\( \frac{\omega}{\lambda g_{\textsf{x}}\epsilon}\geq P_{\max} \r) &=& \int\limits_0^{2^{-\frac{B}{L-1}}}\int\limits_0^\infty\int\limits_{\frac{\theta_{\mathsf{p}}}{\gamma_{\mathsf{p}}}(1+\lambda \gamma_{\max}\tau_2\tau_3 )}^\infty \! f_{g_{\mathsf{p}}}(\tau_1)f_{g_{\textsf{x}}}(\tau_2)f_\epsilon(\tau_3)d\tau_1d\tau_2 d\tau_3\nn\\
&=& \int\limits_0^{2^{-\frac{B}{L-1}}}\int\limits_0^\infty e^{-\frac{\theta_{\mathsf{p}}}{\gamma_{\mathsf{p}}}(1+\lambda \gamma_{\max}\tau_2\tau_3)} f_{g_{\textsf{x}}}(\tau_2)f_\epsilon(\tau_3)d\tau_2 d\tau_3\nn\\
&=& \frac{e^{-\frac{\theta_{\mathsf{p}}}{\gamma_{\mathsf{p}}}}}{\Gamma(L)}\int_0^{2^{-\frac{B}{L-1}}}\int_0^\infty \tau_2^{L-1}e^{-\l(u\tau_3+1\r)\tau_2} d\tau_2 f_\epsilon(\tau_3)d\tau_3 \nn\\
&=& e^{-\frac{\theta_{\mathsf{p}}}{\gamma_{\mathsf{p}}}}\int_0^{2^{-\frac{B}{L-1}}}\frac{f_\epsilon(\tau)}{(1+u\tau)^L }d\tau\label{Eq:App:q}\\
&=& e^{-\frac{\theta_{\mathsf{p}}}{\gamma_{\mathsf{p}}}}(L-1)2^B\int_0^{2^{-\frac{B}{L-1}}}\frac{\tau^{L-2}}{(1+u\tau)^L }d\tau\nn\\
&=& e^{-\frac{\theta_{\mathsf{p}}}{\gamma_{\mathsf{p}}}}(L-1)2^B\int_0^{2^{-\frac{B}{L-1}}}\tau^{L-2}\l[1-Lu\tau+O(\tau^2)\r]d\tau\nn\\
&=& e^{-\frac{\theta_{\mathsf{p}}}{\gamma_{\mathsf{p}}}}\l[1-(L-1)u2^{-\frac{B}{L-1}}\r]+O\l(2^{-\frac{2B}{L-1}}\r).
\label{Eq:Pb}
\end{eqnarray}
The substitution of \eqref{Eq:Pb} into \eqref{Eq:P1:a} gives  \eqref{Eq:P1}.
Next, from \eqref{Eq:Beam:ZF} and \eqref{Eq:Ps}, 
\begin{eqnarray}
 \Pr(P_{\mathsf{s}} < \tau)  &\!=\!&\Pr\(0\leq  \frac{\omega}{\lambda g_{\textsf{x}}\epsilon}\leq \tau \r)\nn\\
&\!=\!&  \Pr\l( \omega \geq 0\r) - \Pr\( \frac{\omega}{\lambda g_{\textsf{x}}\epsilon}\geq \tau \r).\nn
\end{eqnarray}
Using the above equation, \eqref{Eq:P2} is obtained following similar steps as \eqref{Eq:Pb}. This completes the proof. 

\subsection{Proof of Theorem~\ref{Theo:Pout:PerfPwrFb}}\label{App:Outage}
Since the receive SNR at $\mathsf{R_{\mathsf{s}}}$ is  $P_{\mathsf{s}}\tilde{g}_{\mathsf{s}}$, 
\begin{eqnarray}
\Pout &=& \Pr(P_{\mathsf{s}}\tilde{g}_{\mathsf{s}} \leq \theta_{\mathsf{s}}\sigma^2)\nn\\
&=& \int_0^{\frac{\theta_{\mathsf{s}}}{\gamma_{\max}}} \Pr(P_{\mathsf{s}}\leq P_{\max})f_{\tilde{g}_{\mathsf{s}}}(\tau)d\tau + \int_{\frac{\theta_{\mathsf{s}}}{\gamma_{\max}}}^{\infty} \Pr\l(P_{\mathsf{s}}\leq \frac{\theta_{\mathsf{s}}\sigma^2}{\tau}\r)f_{\tilde{g}_{\mathsf{s}}}(\tau)d\tau\nn\\
&\overset{(a)}{=}& 1 - \frac{\Gamma\l(L-1,  \frac{\theta_{\mathsf{s}}}{\gamma_{\max}}\r)}{\Gamma(L-1)} + e^{-\frac{\theta_{\mathsf{p}}}{\gamma_{\mathsf{p}}}} \l[\frac{(L-1)\lambda \theta_{\mathsf{p}}\theta_{\mathsf{s}}}{\gamma_{\mathsf{p}}}\r]2^{-\frac{B}{L-1}}\int_{\frac{\theta_{\mathsf{s}}}{\gamma_{\max}}}^\infty\frac{ \tau^{L-3}}{\Gamma(L-1)}e^{-\tau}d\tau +O\l(2^{-\frac{2B}{L-1}}\r) \nn
\end{eqnarray}
where (a) uses both Lemmas~\ref{Lem:Ps} and \ref{Lem:gs}. The
desired result  follows from the above equation.

\subsection{Proof of Lemma~\ref{Lem:PDF:EspP}}\label{App:PDF:EspP}
Define the random variable  $\kappa=|\bs_1^{\dagger}\bs_2|^2$ where
$\bs_1$ and $\bs_2$ are independent  isotropic vectors in
$\mathds{C}^{L-1}$ with unit norm. The distribution function of $\kappa$ is given
as \cite{MukSabETAL:BeamFiniRateFeed:Oct:03}
\begin{equation}
\Pr\l(\kappa > \tau\r) = (1-\tau)^{L-2}, \quad 0 \leq \tau \leq 1. \label{Eq:PDF:Nu}
\end{equation}
As shown in
\cite{Jindal:MIMOBroadcastFiniteRateFeedback:06},
$\delta$ follows the same distribution as $\kappa\epsilon$. Using
the above results,  the distribution of $\delta$ is readily obtained
as follows:
\begin{eqnarray}
\Pr(\delta\leq t) &=& \Pr(\kappa\epsilon \leq t)\nn\\
&=& \int_0^{2^{-\frac{B}{L-1}}} \Pr\l(\kappa \leq \frac{t}{\tau}\r)f_{\epsilon}(\tau)d\tau \nn\\
&=& \int_0^{t}\Pr\l(\kappa \leq \frac{t}{\tau}\r)f_{\epsilon}(\tau)d\tau+\int_{t}^{2^{-\frac{B}{L-1}}}\Pr\l(\kappa \leq \frac{t}{\tau}\r)f_{\epsilon}(\tau)d\tau\nn\\
&=& 1- \int_{t}^{2^{-\frac{B}{L-1}}}\Pr\l(\kappa > \frac{t}{\tau}\r)f_{\epsilon}(\tau)d\tau\nn\\
&\overset{(a)}{=}& 1- 2^B(L-1)\int_{t}^{2^{-\frac{B}{L-1}}}\l(1-\frac{t}{\tau}\r)^{L-2}\tau^{L-2}d\tau\label{Eq:App:c} \\
&=& 1- 2^B(L-1)\int_{0}^{2^{-\frac{B}{L-1}}-t}\tau^{L-2}d\tau\nn\\
&=& 1- 2^B\l(2^{-\frac{B}{L-1}}-t\)^{L-1}\label{Eq:App:p}
\end{eqnarray}
where \eqref{Eq:App:c} is  obtained by substituting \eqref{Eq:PDF:Nu}. Differentiating
both sides of \eqref{Eq:App:p} gives the desired result.

\subsection{Proof of Lemma~\ref{Lem:Ps:FF}}\label{App:Ps:FF}
Following \eqref{Eq:P1:a} in the proof of 
Lemma~\ref{Lem:Ps}, we can write
\begin{equation}
\Pr(\acute{P}_{\mathsf{s}} = P_{\max}) = \bar{P}_{\textsf{out}} + \Pr\( \frac{\omega}{\lambda g_{\textsf{x}}\delta}\geq P_{\max} \r) \label{Eq:App:r}
\end{equation}
where $ \bar{P}_{\textsf{out}}$  is given in \eqref{Eq:Pout:Prim}.
Using \eqref{Eq:App:q} with  $\epsilon$ replaced by $\delta$,  the last term in \eqref{Eq:App:r} is obtained as 
\begin{eqnarray}
 \Pr\( \frac{\omega}{\lambda g_{\textsf{x}}\delta}\geq P_{\max} \r) &=& e^{-\frac{\theta_{\mathsf{p}}}{\gamma_{\mathsf{p}}}}\int_0^{2^{-\frac{B}{L-1}}}\frac{f_{\delta}(\tau)}{(1+u\tau)^L }d\tau\nn\\
&=& e^{-\frac{\theta_{\mathsf{p}}}{\gamma_{\mathsf{p}}}}(L-1)2^{\frac{B}{L-1}}\int_0^{2^{-\frac{B}{L-1}}}\l(1-2^{\frac{B}{L-1}}\tau\r)^{L-2}[1-Lu\tau + O(\tau^2)]d\tau\label{Eq:App:s}\\
&=& e^{-\frac{\theta_{\mathsf{p}}}{\gamma_{\mathsf{p}}}}\l\{1-(L-1)2^{\frac{B}{L-1}}\int_0^{2^{-\frac{B}{L-1}}}\l(1-2^{\frac{B}{L-1}}\tau\r)^{L-2}[Lu\tau + O(\tau^2)]d\tau\r\}\nn
\end{eqnarray}
\begin{eqnarray}
&=& e^{-\frac{\theta_{\mathsf{p}}}{\gamma_{\mathsf{p}}}}\l[1-L(L-1)2^{-\frac{B}{L-1}} b\int_0^{1}(1-\tau)^{L-2}\tau d\tau + O\l(2^{-\frac{2B}{L-1}}\r)\r]\nn\\
&=& e^{-\frac{\theta_{\mathsf{p}}}{\gamma_{\mathsf{p}}}}\l[1-L(L-1)2^{-\frac{B}{L-1}} b\mathcal{B}(2, L-1)+ O\l(2^{-\frac{2B}{L-1}}\r)\r]\label{Eq:App:a}
\end{eqnarray}
where \eqref{Eq:App:s} applies Lemma~\ref{Lem:PDF:EspP} and 
$\mathcal{B}(\cdot,\cdot)$ represents the beta function. By substituting $\mathcal{B}(x, y) = \frac{\Gamma(x)\Gamma(y)}{\Gamma(x+y)}$ \cite[8.384]{GradRyzhik:Integral:2007} into \eqref{Eq:App:a}
\begin{eqnarray}
 \Pr\( \frac{\omega}{\lambda g_{\textsf{x}}\delta}\geq P_{\max} \r)&=& e^{-\frac{\theta_{\mathsf{p}}}{\gamma_{\mathsf{p}}}}\l[1-L(L-1)2^{-\frac{B}{L-1}} b\frac{\Gamma(2)\Gamma(L-1)}{\Gamma(L+1)}+ O\l(2^{-\frac{2B}{L-1}}\r)\r]\nn\\
&=& e^{-\frac{\theta_{\mathsf{p}}}{\gamma_{\mathsf{p}}}}\l[1-2^{-\frac{B}{L-1}} b+ O\l(2^{-\frac{2B}{L-1}}\r)\r]\label{Eq:App:d}
\end{eqnarray}
where \eqref{Eq:App:d} uses  $\Gamma(L+1) = L!$. Substituting \eqref{Eq:Pout:Prim} and \eqref{Eq:App:d} into  \eqref{Eq:App:r} gives \eqref{Eq:P1:FF}. The desired result in \eqref{Eq:P2:FF} can be obtained following similar steps as given above.

\subsection{Proof of Lemma~\ref{Lem:Delta:P}}\label{App:Delta:P}
Based on the IPC feedback quantization  algorithm in Section~\ref{Section:IPC:OCB:Quant} and for $1\leq n \leq n_0$
\begin{eqnarray}
\Pr(p_{n-1} \leq \eta \leq p_{\mathsf{n}}\mid \gamma_{\mathsf{p}}g_{\mathsf{p}}\geq \theta_{\mathsf{p}}) &=& \frac{\Pr(p_{n-1}\leq P_{\mathsf{s}} < p_{\mathsf{n}})}{\Pr(\gamma_{\mathsf{p}}g_{\mathsf{p}}\geq \theta_{\mathsf{p}})}\nn\\
&=& \frac{(L-1)\theta_{\mathsf{p}}\lambda 2^{-\frac{B}{L-1}}}{\gamma_{\mathsf{p}}\sigma^2}(p_{\mathsf{n}} - p_{n-1}) + O\l(2^{-\frac{2B}{L-1}}\r)\label{Eq:App:e}
\end{eqnarray}
where \eqref{Eq:App:e} uses Lemma~\ref{Lem:Ps}. Combining \eqref{Eq:IPCCbk}, \eqref{Eq:App:e} and $N= 2^A$ gives
\begin{equation}
p_n-p_{n-1} = \frac{\gamma_{\mathsf{p}}\sigma^2}{(L-1)\theta_{\mathsf{p}}\lambda}2^{\frac{B}{L-1}-A} + O\l(2^{-\frac{B}{L-1}}\r),\quad  1\leq n \leq n_0. \label{Eq:App:f}
\end{equation}
The desired result follows from \eqref{Eq:Delta:P} and \eqref{Eq:App:f}.

\subsection{Proof of Lemma~\ref{Lem:Ps:QuantPower}}\label{App:Ps:QuantPower}
The equality in \eqref{Eq:QIPC:Pmax} follows from  the quantized IPC feedback algorithm in Section~\ref{Section:IPC:OCB:Quant}. Based on this algorithm
\begin{eqnarray}
\Pr(\hat{P}_{\mathsf{s}} < \tau)\!\!\! &=&\!\!\! \Pr\l(0\leq\l \lfloor \frac{\omega}{\lambda g_{\textsf{x}}\epsilon}\r\rfloor_{\mathcal{P}}\leq \tau \r)\nn\\
&\leq& \Pr\l(0\leq\frac{\omega}{\lambda g_{\textsf{x}}\epsilon}\leq \tau+\Delta P \r)\label{Eq:App:t}\\
&\leq&  \Pr\(\omega \geq 0 \r) - \Pr\( \frac{\omega}{\lambda g_{\textsf{x}}\epsilon}\geq(\tau+\Delta P) \r)\label{Eq:App:n}
\end{eqnarray}
where \eqref{Eq:App:t} follows from \eqref{Eq:Delta:P}. The
desired result in \eqref{Eq:QIPC:CDF} is obtained using
\eqref{Eq:App:n} and following similar steps as deriving  $\Pr(P_{\mathsf{s}} <
\tau)$ in  Lemma~\ref{Lem:Ps}.

\bibliographystyle{ieeetr}

\end{document}